\documentclass[11pt,draftclsnofoot,onecolumn]{IEEEtran}

\usepackage{amsfonts}
\usepackage{times}
\usepackage{latexsym}
\usepackage{amssymb}
\usepackage{amsmath}
\usepackage{cite}
\usepackage{verbatim}

\newcommand{\bydef}{\triangleq}
\newcommand{\tr}{{\it{tr}}}

\def\bydef{:=}

\def\bb0{{\mathbb{0}}}

\def\bydef{:=}
\def\ba{{\mathbf{a}}}
\def\bb{{\mathbf{b}}}

\def\bg{{\mathbf{g}}}
\def\bh{{\mathbf{h}}}

\def\bn{{\mathbf{n}}}

\def\br{{\mathbf{r}}}
\def\bs{{\mathbf{s}}}
\def\bt{{\mathbf{t}}}

\def\bv{{\mathbf{v}}}
\def\bw{{\mathbf{w}}}
\def\bx{{\mathbf{x}}}
\def\by{{\mathbf{y}}}
\def\bz{{\mathbf{z}}}
\def\b0{{\mathbf{0}}}

\def\bA{{\mathbf{A}}}
\def\bB{{\mathbf{B}}}

\def\bD{{\mathbf{D}}}

\def\bG{{\mathbf{G}}}
\def\bH{{\mathbf{H}}}

\def\bM{{\mathbf{M}}}

\def\bR{{\mathbf{R}}}
\def\bS{{\mathbf{S}}}

\def\bW{{\mathbf{W}}}

\def\bY{{\mathbf{Y}}}
\def\bZ{{\mathbf{Z}}}

\def\bbC{{\mathbb{C}}}

\def\bbE{{\mathbb{E}}}

\def\bbR{{\mathbb{R}}}

\def\cC{\mathcal{C}}

\def\sfH{\mathsf{H}}

\def\bydef{:=}

\def\sf0{{\mathsf{0}}}

\def\Nt{{N_t}}

\usepackage{graphicx}
\usepackage{amssymb}
\usepackage{amsfonts}
\usepackage{amsmath}
\usepackage{latexsym}
\begin{document}
\newtheorem{thm}{Theorem}
\newtheorem{lemma}{Lemma}
\newtheorem{rem}{Remark}
\newtheorem{exm}{Example}
\newtheorem{prop}{Proposition}
\newtheorem{defn}{Definition}
\def\proof{\noindent\hspace{0em}{\itshape Proof: }}
\def\endproof{\hspace*{\fill}~\QED\par\endtrivlist\unskip}
\def\bh{{\mathbf{h}}}
\def\SNR{{\mathsf{SNR}}}
\newcommand{\expeq}{\stackrel{.}{=}}
\newcommand{\expg}{\stackrel{.}{\ge}}
\newcommand{\expl}{\stackrel{.}{\le}}
\title{
Cascaded Orthogonal Space-Time Block Codes for Wireless Multi-Hop Relay Networks
}
\author{Rahul~Vaze and Robert W. Heath Jr. \\
The University of Texas at Austin \\
Department of Electrical and Computer Engineering \\
Wireless Networking and Communications Group \\
1 University Station C0803\\
Austin, TX 78712-0240\\
email: vaze@ece.utexas.edu, rheath@ece.utexas.edu
\thanks{This work was funded in part by Samsung Electronics
and DARPA through IT-MANET grant no. W911NF-07-1-0028.}}

\date{}
\maketitle
\noindent
\begin{abstract}
Distributed space-time block coding is a diversity technique to mitigate
the effects of fading in multi-hop wireless networks, where
multiple relay stages are used by a
source to communicate with its destination.
This paper proposes a new distributed space-time block code called the
cascaded orthogonal space-time block code (COSTBC)
 for the case where the source and destination are equipped with
multiple antennas and each relay stage has one or more single
antenna relays. Each relay stage is assumed to have receive channel
state information (CSI) for all the channels from the source to
itself, while the destination is assumed to have receive CSI for all
the channels. To construct the COSTBC, multiple orthogonal space-time
block codes are used in cascade by the source and each relay stage.
In the COSTBC, each relay stage separates the constellation symbols of
the orthogonal space-time block code sent by the preceding relay
stage using its CSI, and then transmits another orthogonal
space-time block code to the next relay stage. COSTBCs are shown to
achieve the maximum diversity gain in a multi-hop wireless network
with flat Rayleigh fading channels. Several explicit constructions
of COSTBCs are also provided for two-hop wireless networks with two
and four source antennas and relay nodes. It is also shown that
COSTBCs require minimum decoding complexity thanks to the connection
to orthogonal space-time block codes.

\end{abstract}

\section{Introduction}
It is well known that for point-to-point multiple antenna wireless channels,
space-time block codes (STBCs) \cite{Alamouti1998, Tarokh1999a}
improve the bit error rate performance
by introducing redundancy across multiple antennas and time.
Through special designs, STBCs increase the diversity gain, defined
as the negative of the exponent of the signal-to-noise ratio (SNR)
in the pairwise error probability expression at high SNR \cite{Tarokh1999a}.
Recently, the concept of STBC has been extended to wireless networks,
where
the antennas of other nodes in the network (called relays) are used to
construct STBC in a distributed manner to improve the diversity gain
between a particular source and its destination
\cite{Laneman2003,Laneman2004,Jing2004d,Nabar2004,Jing2006a,Belfiore2007,SRajan2006,SRajan2006a}.

%

Prior work on DSTBC
\cite{Laneman2003,Jing2004d,Nabar2004,Jing2006a,Belfiore2007}
considers a two-hop wireless network, where in the first hop the
source transmits the signal to all the relays and in the next hop,
all relays simultaneously transmit a function of the received signal
to the destination.
If a decode and forward (DF) strategy is used, each relay decodes the
incoming signal from the source and then transmits a vector or a matrix
depending on whether it has one or more than one antenna
\cite{Laneman2004, Laneman2003,Damen2007}.
The matrix obtained by stacking all the vectors or matrices
transmitted by the relays is called a DSTBC.
Since each relay decodes the received signal,
the criteria for designing a DSTBC with DF is same as the
 criteria for designing STBCs in point-to-point channels
\cite{Tarokh1999a}.
Due to independent decoding at each relay, however, the diversity gain of DSTBC with DF is limited by the minimum of the diversity gains between the source and all the different relays.

If an amplify and forward (AF) strategy is used,
each relay is only allowed to transmit
a function of the received signal without any decoding, subject to its
power constraint. A DSTBC design is proposed in
\cite{Jing2004d,Jing2006a} using an AF strategy,
where each relay transmits a relay specific unitary transformation of the
received signal. This DSTBC construction, however, is limited to two-hop
wireless networks where each relay is equipped with a single antenna.
It was shown in \cite{Jing2004d,Jing2006a}, that to maximize the diversity
gain, the DSTBC transmitted by all relays using a
unitary transformation should be a full-rank STBC.
Algebraic constructions of maximum diversity gain achieving DSTBC for the
two-hop wireless network are provided in \cite{Oggier2006k,Kiran2006a,EliaDec72005,Jing2007}.

%

Recently, there has been growing interest in multi-hop wireless networks,
where more than two hops are required for a source signal to reach
its destination. Consequently, there is a strong case to construct DSTBCs
that can achieve maximum diversity gain in a large wireless networks
with multiple hops.
Unfortunately, most prior work on constructing DSTBC for
maximizing the diversity gain only considers a two-hop wireless network
\cite{Laneman2004, Laneman2003,Damen2007,Jing2004d,Jing2006a}
and does not readily extends to more than two-hops.

In this paper we design maximum diversity gain
achieving DSTBC's for multi-hop wireless networks.
We assume that the source and the destination terminals have multiple antennas
while the relays in each stage have a single antenna.
We also assume that all the nodes in the network
(source, relays and destination)
can only work in half-duplex mode (cannot transmit and receive at
the same time) and each relay and the destination
has perfect receive channel state information (CSI).

We propose an AF based multi-hop DSTBC, called the cascaded orthogonal
space-time block code (COSTBC), where an orthogonal space-time code
(OSTBC) \cite{Tarokh1999} is used by the source and each relay stage
to communicate with its adjacent relay stage. OSTBCs are considered
because of their single symbol decodable property
\cite{Alamouti1998, Tarokh1999}, i.e. each
constellation symbol of the OSTBC can be separated at the receiver
with independent noise terms. To construct COSTBCs 
the single symbol decodable property
of OSTBC is used by each relay stage to separate the constellation
symbols of the OSTBC transmitted by the preceding stage and transmit
another OSTBC to the next relay stage.

With our proposed COSTBC design, in the first time slot the
source transmits an OSTBC to the first relay stage. Using the single
symbol decodable property of the OSTBC, each relay of the first
relay stage separates the different OSTBC constellation symbols from
the received signal and transmits a codeword vector in the next time
slot, such that the matrix obtained by stacking all the codeword
vectors transmitted by the different relays of the first relay stage
is an OSTBC. These operations are repeated by subsequent relay
stages. It is worth noting that with COSTBC, no signal is decoded
at any of the relays, therefore COSTBC construction
with single antenna relays is equivalent to COSTBC
construction with multiple antenna relays.
Thus without loss of generality in this paper we only consider
COSTBC construction for single antenna relays.
The diversity gain analysis presented in this
paper for COSTBC, however, is very general and applies to the multiple
antenna relay case as well.

We prove that the COSTBCs achieve the maximum diversity gain in two or more
hop wireless networks when CSI is available at each relay and the
destination in the receive mode.
We first show this for a two-hop wireless network
and then using mathematical induction generalize it to the multi-hop case.
We also give an explicit construction of COSTBCs for different
source antennas and relay configurations.
We prove that the COSTBCs have the single symbol decodable
property similar to OSTBCs. We also show that cascading multiple OSTBCs to
construct COSTBC preserves the single symbol decodable property of OSTBCs
and as a result COSTBCs require minimum decoding complexity.

During the preparation of this manuscript we came across three related papers
on DSTBC construction for multi-hop wireless networks
\cite{Oggier2007b,Yang2007a,Sreeram2008}
\footnote{A conference version of our paper was presented in ITA San Diego, Jan. 2008 together with \cite{Sreeram2008}.}.
We briefly review this work and
compare them with the proposed COSTBCs.

Maximum diversity gain achieving DSTBCs are constructed in \cite{Oggier2007b}
for single antenna multi-hop wireless network,
where each node (the source, each relay and the destination)
has single antenna, by extending the AF strategy with unitary transformation
for two-hop wireless networks \cite{Jing2004d}.
It can be shown, however, that the AF strategy with unitary transformation
to construct DSTBC does not extend easily to multi-hop wireless networks
with multiple source or destination antennas. Thus, COSTBC is a more
general solution than the one proposed in \cite{Oggier2007b}.
Moreover, to achieve the maximum diversity gain with the strategy proposed in
\cite{Oggier2007b}, the coding block length,
the time across which coding needs to be done,
is proportional to the product of the number of relay nodes, whereas with
COSTBC it is proportional to the number of relay nodes.
This makes COSTBC more suited for low-latency applications, e.g. voice communication.

The focus of \cite{Yang2007a,Sreeram2008} is on the construction
of DSTBCs that can achieve the optimal diversity multiplexing (DM) tradeoff
\cite{Zheng2003} in a multi-hop wireless network.
In \cite{Yang2007a} a full-duplex multi-hop wireless network
(each node can transmit and
receive at the same time) is considered, whereas
\cite{Sreeram2008} mainly considers a half duplex multi-hop wireless network.
In \cite{Yang2007a} a parallel AF strategy is proposed
which divides the total number of paths from the source to the destination
into non-overlapping groups and transmits an STBC with
non-vanishing determinant property \cite{elia2006c} through each group simultaneously.
It is shown that this strategy achieves the maximum diversity gain and
maximum multiplexing gain points of the optimal DM-tradeoff
in a multi-hop wireless network for some special cases.
An AF strategy similar to delay diversity strategy of \cite{Tarokh1999a}
is proposed in \cite{Sreeram2008} to
achieve the DM-tradeoff for
the half-duplex multi-hop wireless
network where both the source and the destination are equipped with single antenna. In comparison to the strategies of \cite{Yang2007a,Sreeram2008},
COSTBC only achieves the maximum diversity gain
and not the maximum multiplexing gain. Due to the use
of OSTBCs, however, the decoding complexity of COSTBC
is significantly less than the strategies of \cite{Yang2007a,Sreeram2008}
where STBCs with high decoding complexity are used.
Thus COSTBCs are more suited for practical implementation
than the strategies of \cite{Yang2007a, Sreeram2008}.

{\it Notation:}
Let ${\bA}$ denote a matrix, ${\bf a}$ a vector and $a_i$ the
$i^{th}$ element of ${\bf a}$. The $i^{th}$ eigenvalue of $\bA$
is denoted by $\lambda_i(\bA)$ and the maximum and minimum eigenvalue of
$\bA$ by $\lambda_{max}(\bA)$ and $\lambda_{min}(\bA)$, respectively,
if the eigenvalues of $\bA$ are real.
The determinant and trace of matrix  ${\bf A}$ are denoted by
$\det({\bA})$ and $\tr({\bA})$, while $\bA^{\frac{1}{2}}$ denotes
the element wise square root of matrix ${\bf A}$ with all non-negative entries.
The field of real and complex numbers is denoted by $\bbR$ and $\bbC$,
respectively.
The space of $M\times N$ matrices with complex entries is denoted by
${\bbC}^{M\times N}$.
The Euclidean norm of a vector $\bf a$ is denoted by $|\ba|$.
An $m\times m$ identity matrix is denoted by ${\bf I}_m$ and ${\bf 0}_m$ is
as an all zero $m\times m$ matrix.
The superscripts $^T, ^*, ^{\dag}$ represent the transpose, transpose
conjugate and element wise conjugate.
For matrices $\bA,\bB$ by
$\bA \le \bB, \bA, \bB \in {\bbC}^{m \times m}$ we mean
$\bx\bA\bx^{*} \le \bx\bB\bx^{*}, \ \forall \bx\in {\bbC}^{1\times m}$.
The expectation of function $f(x)$ with respect to $x$
is denoted by ${\bbE}_{\{x\}}f(x)$.
The maximum and minimum value of the set
$\{a_1,a_2,\ldots,a_m\}$ where $a_i \in \bbR, \ i=1,2,\ldots,m$
are denoted by $\max\{a_1,a_2,\ldots,a_m\}$ and $\min\{a_1,a_2,\ldots,a_m\}$.
A circularly symmetric complex Gaussian random variable $x$ with zero mean and
variance $\sigma^2$
is denoted as $x \sim {\cal CN}(0,\sigma)$.
We use the symbol $\expeq$ to represent exponential equality i.e.,
let $f(x)$ be a
function of $x$, then  $f(x) \expeq x^a$ if $\lim_{x\rightarrow \infty}\frac{\log(f(x))}{\log x} = a$ and similarly $\expl$ and $\expg$ denote the exponential
less than or equal to and greater than or equal to relation, respectively.
To define a variable we use the symbol $\bydef$.

{\it Organization:} The rest of the paper is organized as follows.
In Section \ref{sec:sys}, we describe the system model for the multi-hop
wireless network and review the key assumptions.
In Section \ref{sec:costbc}, COSTBC construction
is described.
A diversity gain analysis of COSTBCs is presented in Section \ref{sec:2hop}
for the $2$-hop case and generalized to $N$-hop network case in
Section \ref{sec:multihop}.
In Section \ref{sec:code}, explicit constructions
of COSTBCs are provided which achieve maximum diversity gain
for different number of source antenna and relay node configurations.
Some numerical results are provided
in Section \ref{sec:simulation}.
Final conclusions are made in Section \ref{sec:conc}.

\section{System Model}
\label{sec:sys}
\begin{figure}
\centering
\includegraphics[width=6.5in]{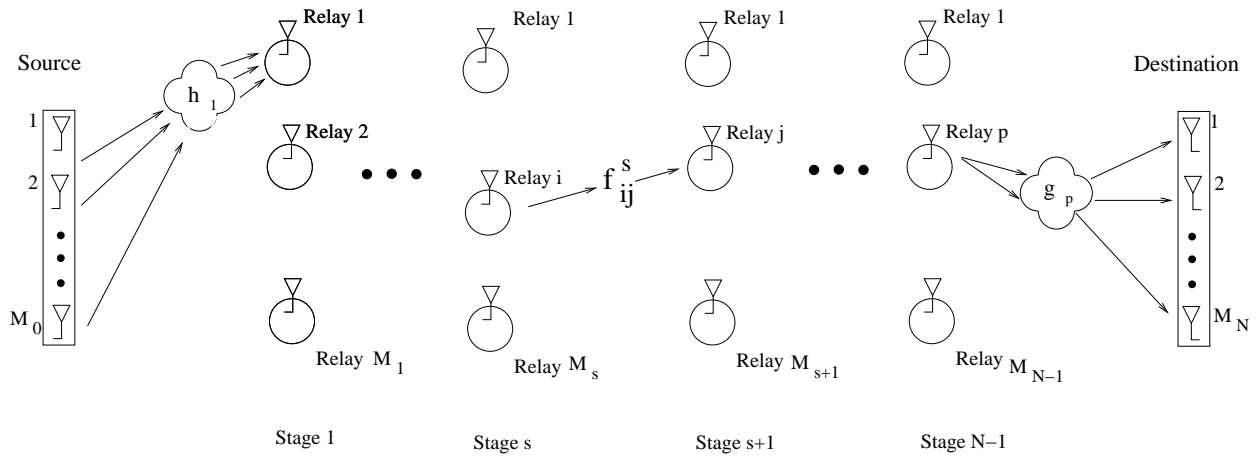}
\caption{System Block Diagram of a N-hop Wireless Network}
\label{blkdiag}
\end{figure}

We consider a multi-hop wireless network where a source terminal
with $M_0$ antennas wants to communicate with a destination terminal
with $M_N$ antennas via $N-1$ relay stages as shown in Fig.
\ref{blkdiag}. Each relay in any relay stage has a single antenna;
$M_n$ denotes the number of relays in the $n^{th}$ relay stage. It
is assumed that the relays do not generate their own data and only
operate in half-duplex mode. A half-duplex assumption is made since
full-duplex nodes are difficult to realize in practice. Similar to the model
considered in \cite{Yang2007a}, we assume that
any relay of relay stage $n$ can only receive the signal from any relay of
relay stage $n-1$, i.e. we consider a directed multi-hop wireless network.
In a practical system this assumption can be realized by allowing every
third relay stage to be active (transmit or receive) at the same time.
To keep the relay functionality and relaying strategy simple we do not allow relay nodes to
cooperate among themselves. We assume that there is no direct path
between the source and the destination. This is a reasonable
assumption for the case when relay stages are used for coverage
improvement and the signal strength on the direct path is very weak.
Throughout this paper we refer to this multi-hop wireless network
with $N-1$ relay stages as an $N$-hop network.

As shown in Fig. \ref{blkdiag}, the channel between the source and the $i^{th}$
relay of the first stage of relays is denoted by
$ \bh_i = [h_{1i} \ h_{2i} \ \ldots \  h_{M_{0}i}]^T, \ i=1,2,\ldots,M_1$,
between the $j^{th}$
relay of relay stage $s$ and the $k^{th}$ relay of relay stage $s+1$ by $f^{s}_{jk}
,\ s=0,1,\ldots,N-2, \ j=1,2,\ldots,M_s,\ k=1,2,\ldots,M_{s+1}$
and the channel between the relay stage $N-1$
and the $\ell^{th}$ antenna of the destination by
$\bg_{\ell}= [ g_{1\ell} \ g_{2\ell} \ \ldots \  g_{M_{N-1}\ell}] ^T, \ \ell=1,2,\ldots,M_{N}$.
We assume that $\bh_i \in {\mathbb C}^{M_{0}\times 1}$, $f^{s}_{jk}\in {\mathbb C}^{1\times 1}$,
$\bg_l \in {\mathbb C}^{M_{N-1}\times 1}$ with
independent and identically distributed
 (i.i.d.) ${\cal CN}(0,1)$ entries for all $i,j,k,\ell,s$.
We assume that the $m^{th}$ relay of $n^{th}$ stage
knows $\bh_i,  f^s_{jk},\ \forall \ i, \ j , \ k, s=1,2,\ldots,n-2, \ f^{n-1}_{jm} \ \forall j$
and the destination knows $\bh_i, f^{s}_{jk}, \bg_l, \
\forall \ i,j,k,l,s$.
We further assume that all these channels are
frequency flat and block fading, where
the channel coefficients remain constant in a block of time
duration $T_c$ and
change independently from block to block.
We assume that the $T_c$ is at least
$\max\{M_0, M_1, \ldots, M_{N-1}\}$.

\label{sec:sys}
\begin{figure}
\centering
\includegraphics[width= 7in]{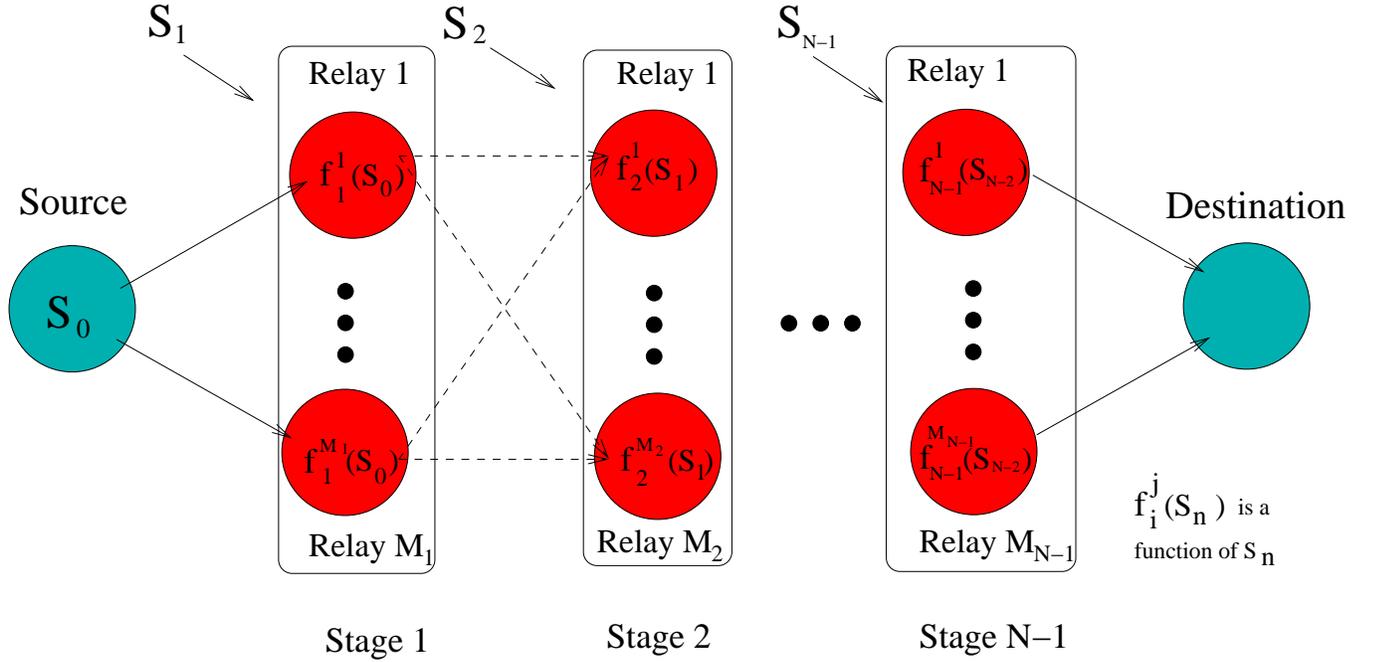}
\caption{An Illustration Of The DSTBC Design Problem}
\label{stcoding}
\end{figure}

\subsection{Problem Formulation}
\begin{defn}(STBC) \cite{Sethuraman2003}
A rate-$L/T$ $T \times \Nt $ {\bf design} ${\bD}$ is a $T \times \Nt $
 matrix with entries that are complex
linear combinations of $L$ complex variables
$s_1, s_2, \ldots, s_L$
and their complex conjugates.
A rate-$L/T$ $T \times \Nt $ STBC ${\bf S}$ is a set of $T \times \Nt$ matrices that are
obtained by allowing the $L$ variables $s_1, s_2, \ldots, s_L$
of the rate-$L/T$ $T \times \Nt$ design $\bD$ to take values from a finite
subset ${\bbC}^f$ of the complex field ${\bbC}$. The cardinality of
${\bf S} = |{\bbC}^f|^{L}$, where $|{\bbC}^f|$ is the cardinality of ${\cC}$.
We refer to $s_1, s_2, \ldots, s_L$ as the constituent symbols of the
STBC.
\end{defn}
\begin{defn}
A DSTBC ${\cal C}$ for a $N$-hop network is a
collection of codes $\left\{\bS_0, \bS_1, \ldots, \bS_{N-1}\right\}$,
where $\bS_0$ is the STBC transmitted by the source and
$\bS_n = [{\bf f}^1_n({\bS_{n-1}})  \ldots {\bf f}^{M_n}_n({\bS_{n-1}})]$ is
the STBC transmitted by relay stage $n$, where ${\bf f}^j_n({\bS_{n-1}})$ is
the vector transmitted by the $j^{th}$ relay of
stage $n$ which is a function of
${\bS_{n-1}},  \ j=0,\ldots, M_n, n=1,\ldots, N-1$.
An example of a DSTBC is illustrated in Fig. \ref{stcoding}.
\end{defn}
\begin{defn}
The diversity gain \cite{Tarokh1999a, Jing2004d} of a DSTBC ${\cal C}$
is defined as
\[ d_{\cal C} = -\lim_{E\rightarrow \infty}\frac{\log{P_e\left(E\right)}}{\log{E}},
\] $P_e\left(E\right)$ is the pairwise error probability (PEP) using coding
strategy ${\cal C}$, and $E$ is the sum of the transmit power used by each
node in the network.
\end{defn}
The problem we consider is in this paper is to design DSTBCs that
achieve the maximum diversity gain in a $N$-hop network.
To identify the limits on the maximum possible
diversity gain in a $N$-hop network, an upper bound on the diversity gain achievable with any DSTBC is presented next.

\begin{thm}
\label{upbound}
The diversity gain $d_{\cal C}$ of DSTBC ${\cal C}$ for
an $N$-hop network is upper bounded by
\[\min\left\{M_nM_{n+1}\right\} \ \ n=0,1,\ldots,N-1.\]
\end{thm}
\begin{proof}
Let $d_{\cal C}$ be the diversity gain of coding strategy ${\cal C}$ for
an $N$-hop network.
Let $d_n$ be the
diversity gain of the best possible DSTBC ${\cal C}^{opt}$
that can be used between relay stage
$n$ and $n+1$ when all the relays
in relay stage $n$ and relay stage $n+1$ are allowed to collaborate,
respectively, and
the source message is known to all the relays of relay stage $n$ without
any error and all the relays of the relay stage $n+1$
can send the received signal to the destination error free. Then, clearly,
$d_{\cal C} \le d_n$.
Since the channel between the relay stage $n$ and $n+1$
is a multiple antenna channel with $M_n$ transmit
and $M_{n+1}$ receive antennas,
$d_n\le M_nM_{n+1}$.
Hence $d_{\cal C}\le M_nM_{n+1}$.
Since this is true for
every $n=0,1,\ldots,N-1$, it follows that $d_{\cal C} \le \min\{M_nM_{n+1}\}, \ n=0,1,\ldots,N-1$.
\end{proof}

Thus, Theorem \ref{upbound} implies that the maximum diversity gain
achievable in a $N$-hop network is equal to the minimum of the maximum
diversity gain achievable between any two relay stages, when all the relays
in each relay stage are allowed to collaborate.
In our system model we do not allow any cooperation between relays,
and hence designing a DSTBC that achieves the diversity gain upper bound without any cooperation is
difficult.

For the case of $2$-hop networks, DSTBCs have been proposed to
achieve the maximum diversity gain \cite{Jing2004d, Jing2006a}.
It is worth noting that designing DSTBCs that achieve the maximum diversity
gain in a
$N$-hop network is a difficult problem. The difficulty is two-fold:
proposing a ``good" DSTBC and analyzing its diversity gain.
In the next section we describe our novel COSTBC construction
and prove that it achieves the maximum diversity gain in a $N$-hop network.
As it will be clear in the next section, using OSTBCs to construct COSTBC
simplifies the diversity gain analysis, significantly.

\section{Cascaded Orthogonal Space-Time Code}
\label{sec:costbc}
In this section we introduce the COSTBC design
for a $N$-hop network.
Before introducing COSTBC we need the following definitions.
\begin{defn}
With $T\ge \Nt$, a rate $L/T$ $T \times \Nt $ STBC ${\bf S}$ is called full-rank or fully-diverse
or is said to achieve maximum diversity gain
if the difference of any two matrices $\bM_1, \bM_2 \in {\bf S}$ is full-rank,
\[\min_{\bM_1\ne \bM_2, \ \bM_1,\bM_2 \in {\bf S}}rank(\bM_1- \bM_2) =\Nt.\]
\end{defn}

\begin{defn}(OSTBC)
A rate-$L/K$  $K\times K$ STBC ${\bf S}$ is called an
orthogonal space-time block code (OSTBC) if the design $\bD$ from which it is
derived is orthogonal i.e.
${\bD}{\bD}^{*} = (|s_1|^2 + \ldots + |s_L|^2) {\bf I}_K$.
\end{defn}

\begin{defn}
Let ${\bf S}$ be a rate-$L/K$  $K\times K$ STBC. Then, using CSI,
if each of the constituent symbols $s_i, \ i=1,\ldots,L$ of ${\bf S}$ can be 
separated/decoded independently of $s_j$ $\forall i\ne j \ i,j =1,\ldots,L$
with independent noise terms, then
${\bf S}$ is called a single symbol decodable STBC.
\end{defn}
\begin{rem}
OSTBCs are single symbol decodable STBCs \cite{Tarokh1999}.
\end{rem}
With these definitions we are now ready to describe
COSTBC for a $N$-hop network.

COSTBC is a DSTBC where each $\bS_n, \ n=0,1,\ldots,N-1$ is an OSTBC.
Thus, with COSTBC the source transmits a rate-$L/M_0$ $M_0\times M_0$ OSTBC 
$\bS_0$
in time slot of duration $M_0$. How to construct OSTBCs
$\bS_n, n=1,\ldots,N-1$ is detailed in the following.
Let $\bS_0$ be a rate-$L/M_0$ $M_0\times M_0$ OSTBC transmitted by the source
OSTBC ${\bf S}_0 \in \bbC^{M_0\times M_0}$ to all the relays
of relay stage $1$. Then the received signal $\br_k^1 \in \bbC^{M_0 \times 1}$ at relay $k$ of relay stage $1$ can be written as
\begin{equation}
\br_k^1 =  \sqrt{E_0} \bS_0\bh_{k} + \bn_k^1
\end{equation}
where ${\bbE}tr(\bS_0^{*}\bS_0)=
M_0$ and $E_0$ is the power transmitted
by the source at each time instant. The noise $\bn_k^1$ is the $M_0\times 1$ spatio-temporal
white complex Gaussian noise independent across relays with $\bbE\bn_k^1\bn_k^{1*} = {\bf I}_{M_0}$.
Since ${\bf S}_0$ is an OSTBC, using CSI, the received signal $\br_k^1$
can be transformed into ${\tilde \br}^1_k\in \bbC^{L\times1}$, where
\begin{equation}
\label{ssd}
{\tilde \br}_k^1 =  \sqrt{E_0}
\underbrace{\left[\begin{array}{cccc}
\sum_{m=1}^{M_0}|h_{mk}|^2 &  0 & 0 \\
0 & \ddots & 0 \\
0 &  0&  \sum_{m=1}^{M_0}|h_{mk}|^2\\
\end{array}\right]}_{\sfH}\bs + {\tilde \bn}_k^1
\end{equation}
and $\bs = [s_1, s_2, \ \ldots, \ s_L]^T$ is the vector of the constituent
symbols of the
OSTBC $\bS_0$, ${\sfH }$ is an $L \times L$ matrix and
${\tilde \bn}_k^1$ is an $L \times 1$ vector with entries that are uncorrelated
and ${\cal CN}(0, M_0)$ distributed.
This property is illustrated in the Appendix
\ref{ssdcascala} for the case of the Alamouti code
\cite{Alamouti1998} which is an OSTBC for $M_0=2$. Then we normalize
${\tilde r}_k^1$ by $\sfH^{-\frac{1}{2}}$ to obtain ${\hat r}_k^1$, where
\begin{eqnarray}
\label{scaledssd}\nonumber
{\hat \br}_k^1 &\bydef& \sfH^{-\frac{1}{2}}{\tilde \br}_k^1 \\
 & = & \sqrt{E_0}
\sfH^{\frac{1}{2}}\bs + \underbrace{\sfH^{-\frac{1}{2}}{\tilde \bn}_k^1}_{{\hat \bn}_k^1},
\end{eqnarray}
where ${\hat \bn}_k^1$ is an $L \times 1$ vector with entries
that are uncorrelated and ${\cal CN}(0, 1)$ distributed.

Then, in the second time slot of duration $M_1$, relay $k$ of relay stage $1$
transmits $\bt_k^1$, constructed from the signal (\ref{scaledssd})
\begin{equation}
\label{relaytx}
\bt_k^1 = \sqrt{\frac{E_1M_1}{L\gamma}}\left(\bA_k\hat{\br}_k^1 + \bB_k\hat{\br}^{1\dag}_k\right),
\end{equation}
where $\gamma = \bbE\hat{\br}_k^{1*}\hat{\br}_k^1$ to ensure that the
average power transmitted by each relay at any time instant is $E_1$, i.e.
\[\bbE\left(\bt_k^{1}\right)^{\dag}\left(\bt_k^1\right) = E_1\]
and $\bA_k$, $\bB_k$ are $M_1\times L$ matrices such that
\begin{eqnarray}
\label{hurrad}
\nonumber
\bA_k^{*}\bB_k & =& - \bB_k^{*}\bA_k  \ \ \text{and}\\
\tr{\left(\bA_k^{*}(l)\bA_k(l) +\bB_k^{*}(l)\bB_k(l)\right)} &=& 1  \ \ \forall \ k=1,2\ldots,M_1, \ l=1,2,\ldots L,
\end{eqnarray}
where $\bA_k(l)$ and $\bB_k(l)$ denote the $l^{th}$ column of
$\bA_k$ and $\bB_k$, respectively and
\[\bS_1 \bydef
[\bA_1\bs +\bB_{1}\bs^{\dag} \ldots  \bA_{M_1}\bs+ \bB_{M_1}\bs^{\dag}]\] is
an OSTBC.

Under these assumptions, the $M_1\times 1$ received signal at the
$i^{th}$ relay of relay stage $2$ is
\begin{eqnarray*}
\label{recdestant}
\by_i &=& \sum_{k=1}^{M_1}\bt_kg_{ki} + \bz_i \\
&=&\sqrt{\frac{E_0E_1M_1}{L\gamma}}
\underbrace{[\bA_1\bs+\bB_{1}\bs^{\dag} \ \bA_2\bs+\bB_{2}\bs^{\dag} \  \ldots \
\bA_{M_1}\bs+\bB_{M_1}\bs^{\dag} ]}_{\bS_1} {\hat \bH}^{\frac{1}{2}}\bg_i \\ & &  +
\sqrt{\frac{E_1M_1}{L\gamma}}[\bA_1\hat{\bn}^1_1 + \bB_{1}\hat{\bn}^{1\dag}_{1} \
 \ \ldots \
\bA_{M_1}\hat{\bn}^1_{M_1} + \bB_{M_1}\hat{\bn}^{1\dag}_{M_1} ]\bg_i + \bz_i
\end{eqnarray*}
for $i = 1,2, \ldots M_2$, where
$\bz_i$ is the $M_1\times 1$
spatio-temporal white complex Gaussian noise independent across
$M_2$ receive antennas with i.i.d. ${\cal CN}(0, 1)$ entries and
\[{\hat \bH}^{\frac{1}{2}}=
\left[\begin{array}{cccc}
\sqrt{\sum_{m=1}^{M_0}|h_{m1}|^2} & 0 & 0 & 0 \\
0 & \sqrt{\sum_{m=1}^{M_0}|h_{m2}|^2}  & 0 & 0 \\
0 & 0 & \ddots & 0 \\
0 & 0 & 0&  \sqrt{\sum_{m=1}^{M_0}|h_{mM_1}|^2}\\
\end{array}\right].\]

Thus, an OSTBC $\bS_1$ is transmitted by relay stage $1$ to the relay stage 
$2$ in a distributed manner.
To construct the COSTBC, the strategy of transmitting an OSTBC
from relay stage $1$ is repeated at each relay stage,
i.e. each relay of relay stage $n$ transforms the
received signal as in (\ref{scaledssd}) for the OSTBC transmitted from the
relay stage $n-1$ and transmits an OSTBC in time duration $M_n$ using
$A_k,B_k, k=1,\ldots, M_n$ together with all the other relays in relay stage
$n$ to the relay
stage $n+1$. The power used up at each relay of relay stage $n$ is
$E_n$ such that $E_0+\sum_{n=1}^{N-1} M_nE_n = E$, where $E$ is the
total power available in the network. In the $N^{th}$ time slot of
duration $M_{N-1}$ the receiver receives an OSTBC from relay stage
$N-1$.

The properties of the COSTBC are summarized in the next two Theorems.
\begin{thm}
COSTBCs achieve the maximum diversity
gain in a $N$-hop network given by Theorem \ref{upbound}.
\end{thm}

We prove this Theorem in the next two sections.
We start with the $N=2$ case and show that the COSTBCs achieve the
maximum diversity gain for $2$-hop network in Section \ref{sec:2hop}
and then generalize the
result to an arbitrary $N$-hop network using mathematical induction in Section
\ref{sec:multihop}.
\begin{thm}
COSTBCs are single symbol decodable STBCs.
\end{thm}

The Theorem is proved in Appendix \ref{ssdcascala} for a special case 
of $M_n=2, \ n=1,2,\ldots,N-1$ and 
in Appendix \ref{ssdcostbc} for the general case.
Recall that
with COSTBCs, OSTBCs are transmitted in cascade by each relay stage, thus
the single symbol decodable property of the COSTBCs implies that by cascading OSTBCs, the
single symbol decodable property
of OSTBC is preserved.
We also make use of the single symbol decodable property of COSTBC to show that it achieves
the maximum diversity gain for $N$-hop networks.

\section{Diversity Gain Analysis of COSTBC For $2$-Hop Network}
\label{sec:2hop}
In this section we prove that the COSTBCs achieve the maximum diversity gain in a
$2$-hop network.
\begin{thm}
\label{2hopdiv}
COSTBCs achieve a diversity gain of $\min\{M_0M_1, \ M_1M_2\}$ in a $2$-hop
network.
\end{thm}
\begin{proof}
Using a COSTBC in a $2$-hop
network, from (\ref{rxreceive}), the received signal at the $i^{th}$ antenna of
destination is
\begin{equation}
\by_i = \sum_{k=1}^{M_1}\bt_kg_{ki} + \bz_i.
\end{equation}

Then the received signal $\bY \bydef [\by_1 \ldots \by_{M_2}]$ at the destination,
received in time slots $M_0+1$ to $M_0+M_1+1$ can be written as
\begin{eqnarray*}
\label{rxreceive}
  \bY & =&  \sqrt{\frac{E_0E_1M_1}{L\gamma}}
\underbrace{[\bA_1\bs+\bB_{1}\bs^{\dag} \ \bA_2\bs+\bB_{2}\bs^{\dag} \  \ldots \
\bA_{M_1}\bs+\bB_{M_1}\bs^{\dag} ]}_{\bS_1} {\hat \bH}\bG \\ & &  +
\underbrace{\sqrt{\frac{E_1M_1}{L\gamma}}[\bA_1\hat{\bn}^1_1 + \bB_{1}\hat{\bn}^{1\dag}_{1} \
 \ \ldots \
\bA_{M_1}\hat{\bn}^1_{M_1} + \bB_{M_1}\hat{\bn}^{1\dag}_{M_1} ]\bG + \bZ}_{\bW}
\end{eqnarray*}
where
\[\bG = [\bg_1 \ldots \bg_{M_2}] = \left(\begin{array}{cccc}
g_{11} & g_{12} & \ldots  & g_{1M_2} \\
g_{21} & g_{22} & \ldots  & g_{2M_2} \\
\vdots  & \vdots & \ldots & \vdots \\
g_{M_11} & g_{M_12} & \ldots  & g_{M_1M_2}\\
\end{array}\right)\]
and the noise $\bZ = [\bz_1 \ \bz_2 \ \ldots \ \bz_{M_2}]$.

Concisely, we can write
\begin{equation}
\label{origchannel}
\bY = \sqrt{\frac{E_0E_1M_1}{L\gamma}}\bS_1{\hat \bH}^{\frac{1}{2}}\bG + \bW.
\end{equation}
With channel coefficients $\bh_k$ and $\bg_k$ known at the receiver
$\forall \ k =1,2,\dots, M_1$, $\bW$ is Gaussian distributed
with an all zero mean vector and entries of $Y$ are Gaussian distributed.
Moreover it can be shown that any two rows of $\bY$ are uncorrelated and hence
independent.

Using the definition of $\bA_k$ and $\bB_k$ and the fact that
$\hat{\bn}^1_k$ is $L\times 1$
vectors with ${\cal CN}\left(0, \sum_{i=1}^{M_0}|h_{im}|^2\right)$
entries $\forall \ k=1,2,\ldots, M_1$, it can be shown that
the covariance matrix $\bR_{\bW}$ of each row of $\bW$ is
\[\bR_{\bW} = \frac{E_1M_1}{\gamma}\bG^{*}\bG + {\bf I}_{M_2}.\]

Defining $\Phi = \hat{\bH}^{\frac{1}{2}}\bG$,
\begin{eqnarray*}
P\left(\bY|{\bS_1}{\hat \bH},\bG\right) &=&
\prod_{t=1}^{M_1}P\left([\bY]_t|{\bS_1}{\hat \bH},\bG\right)\\
& =& \left(\frac{1}
{{(2\pi)}^{M_2}
\det\left({\bR_{\bW}}\right)
}\right)^{\frac{M_1}{2}}
e^{
-tr
\left(
\left[\bY - \sqrt{\frac{E_0E_1M_1}{L\gamma}}{\bS_1}\Phi\right]
\bR_{\bW}^{-1}
\left[\bY - \sqrt{\frac{E_0E_1M_1}{L\gamma}}{\bS_1}\Phi\right]^{*}
\right)}
\end{eqnarray*}
where $P(\bY|{\bS_1},{\hat \bH},\bG)$ is the conditional probability of
$\bY$ given ${\bS_1},\ {\hat \bH}, \ \bG$ and
$P([\bY]_t|{\bS_1},{\hat \bH},\bG)$ is the conditional probability of
$t^{th}$ row of $\bY$ given ${\bS_1},\ {\hat \bH}, \ \bG$.
Assuming  ${\bS_1}_l$ is the transmitted codeword, then for any
$\lambda > 0$, the PEP $P\left({\bS_1}_l \rightarrow {\bS_1}_m\right)$
of decoding a codeword ${\bS_1}_m, \ m\ne l$,
has the Chernoff bound \cite{V.1968}
\[P\left({\bS_1}_l \rightarrow {\bS_1}_m\right) \le {\bbE}_{\{{\hat \bH},\bG,\bW\}}e^{\lambda
\left(\log{P(\bY|{\bS_1}_l,{\hat \bH},\bG)} - \log{P(\bY|{\bS_1}_m,{\hat \bH},\bG)}\right)}.\]
Since ${\bS_1}_l$ is the correct transmitted codeword,
\[\bY = \sqrt{\frac{E_0E_1M_1}{L\gamma}}{\bS_1}_l\Phi + \bW \]
and
\begin{eqnarray*}
\log{P(\bY|{\bS_1}_l,{\hat \bH},\bG)} - \log{P(\bY|{\bS_1}_m,{\hat \bH},\bG)} & = & -\tr\left[
\frac{E_0E_1M_1}{L\gamma}\left({\bS_1}_l - {\bS_1}_m\right)\Phi \bR_{\bW}^{-1}\Phi^{*}\left({\bS_1}_l - {\bS_1}_m\right)\right. \\
& & +   \sqrt{\frac{E_0E_1M_1}{L\gamma}}\left({\bS_1}_l - {\bS_1}_m\right)\Phi \bR_{\bW}^{-1}\bW^{*}\\
& &  +  \left. \sqrt{\frac{E_0E_1M_1}{L\gamma}}\bW \bR_{\bW}^{-1}\Phi^{*}\left({\bS_1}_l - {\bS_1}_m\right)^{*}\right].
\end{eqnarray*}

Therefore,
\begin{eqnarray}
\label{pep1}
&&P\left({\bS_1}_l \rightarrow {\bS_1}_m\right)  \nonumber\\
&&\le {\bbE}_{\{{\hat \bH},\bG\, \bW\}}e^{-\lambda \tr\left(\frac{E_0E_1M_1}{L\gamma}\left({\bS_1}_l - {\bS_1}_m\right)\Phi \bR_{\bW}^{-1}\Phi^{*}\left({\bS_1}_l - {\bS_1}_m\right)
+   \sqrt{\frac{E_0E_1M_1}{L\gamma}}\left({\bS_1}_l - {\bS_1}_m\right)\Phi \bR_{\bW}^{-1}\bW^{*}
 +   \sqrt{\frac{E_0E_1M_1}{L\gamma}}\bW \bR_{\bW}^{-1}\Phi^{*}\left({\bS_1}_l - {\bS_1}_m\right)^{*}\right)} \nonumber\\
&& \le
\bbE_{\{{\hat \bH},\bG\}}
e^{
-\lambda(1-\lambda)\frac{E_0E_1M_1}{L\gamma}
\tr\left(\left({\bS_1}_l - {\bS_1}_m\right)
\Phi \bR_{\bW}^{-1}\Phi^{*}
\left({\bS_1}_l - {\bS_1}_m\right)^{*}
\right)
} \nonumber\\
&&
\int
{
\frac
{
e^{
-\tr
\left(
\left(\lambda \sqrt{\frac{E_0E_1M_1}{L\gamma}}\left({\bS_1}_l - {\bS_1}_m\right)\Phi+\bW\right)
\bR_{\bW}^{-1}
\left(\lambda \sqrt{\frac{E_0E_1M_1}{L\gamma}}\left({\bS_1}_l - {\bS_1}_m\right)\Phi+\bW\right)^{*}
\right)
}
}
{\pi^{M_2T}\det^{-1}(\bR_{\bW})}d\bW}  \nonumber\\
&&\le
{\bbE}_{\{{\hat \bH},\bG\}}
e^{-\lambda(1-\lambda)\frac{E_0E_1M_1}{L\gamma}
\tr\left(\left({\bS_1}_l - {\bS_1}_m\right)\left({\bS_1}_l - {\bS_1}_m\right)^{*}\Phi \bR_{\bW}^{-1}\Phi^{*}\right)}.
\end{eqnarray}
Clearly $\lambda = \frac{1}{2}$ maximizes $\lambda(1-\lambda)$ for
$\lambda > 0$, and
therefore minimizes the above expression and it follows that
\begin{equation}
\label{pep1}
P\left({\bS_1}_l \rightarrow {\bS_1}_m\right)
\le
{\bbE}_{\{{\hat \bH},\bG\}}
e^{-\frac{E_0E_1M_1}{4L\gamma}
tr\left[\left({\bS_1}_l - {\bS_1}_m\right)\left({\bS_1}_l - {\bS_1}_m\right)^{*}\Phi \bR_{\bW}^{-1}\Phi^{*}\right]}.
\end{equation}

The difficulty in evaluating the expectation in (\ref{pep1}) is the fact that the
noise covariance matrix $\bR_{\bW}$ is not diagonal. To simplify the PEP
analysis we use an upper bound on the eigenvalues of $\bR_{\bW}$, derived in the next lemma.
\begin{lemma}
\label{eigsupper}
$\bR_{\bW} \le \left(1 + \frac{E_1M_1^2}{\gamma}\lambda_{max}\left(\frac{\bG^{*}\bG}{M_1}\right)\right){\bf I}_{M_2}$.
\end{lemma}
\begin{proof}
Recall that $\bR_{\bW} = {\bf I}_{M_2} + \frac{E_1M_1}{\gamma}\bG^{*}\bG$.
Thus the eigenvalues $\lambda_i\left(\bR_{\bW}\right) = 1 +
\frac{E_1M_1^2}{\gamma}\lambda_i\left(\frac{\bG^{*}\bG}{M_1}\right), \ \forall \ i=1,2\ldots,M_2$ and clearly
\[\bR_{\bW} \le \left(1 + \frac{E_1M_1^2}{\gamma}\lambda_{max}\left(\frac{\bG^{*}\bG}{M_1}\right)\right){\bf I}_{M_2}.\]
\end{proof}
From here on in this paper we refer to
$\lambda_{max}\left(\frac{\bG^{*}\bG}{M_1}\right)$
as $\lambda_{\bG}$ for notational simplicity.
Using Lemma \ref{eigsupper}, (\ref{pep1}) simplifies to
\begin{equation}
\label{peppowerall}
P\left({\bS}_{1l} \rightarrow {\bS}_{1m}|\lambda_{\bG}=\lambda_0\right)
\le
{\bbE}_{\{{\hat \bH},\bG\}}
e^{-\frac{E_0E_1M_1}{4M_0L\left(\gamma+E_1M_1^2\lambda_0\right)}
tr\left[\left({\bS}_{1l} - {\bS}_{1m}\right)\left({\bS}_{1l} - {\bS}_{1m}\right)^{*}\Phi\Phi^{*}\right]}, 
\end{equation}
where 
\begin{equation*}
P\left({\bS}_{1l} \rightarrow {\bS}_{1m}\right)
= \bbE_{\{\lambda_{\bG}\}}P\left({\bS}_{1l} \rightarrow {\bS}_{1m}|\lambda_{\bG}=\lambda_0\right).
\end{equation*}

Recall that there is a power constraint of $E_0+E_1M_1 =E$. Therefore
to minimize the upper bound on the PEP (\ref{peppowerall}), the term
$\frac{E_0E_1M_1}{4M_0L\left(\gamma+E_1M_1^2\lambda_0\right)}$ should be maximized
over $E_0, E_1$ satisfying the power constraint.
The optimal values of $E_0$ and $E_1$ to maximize
$\frac{E_0E_1M_1}{4M_0L\left(\gamma+E_1M_1^2\lambda_0\right)}$ can be found
explicitly, however, they can complicate the diversity gain analysis.
To simplify the diversity gain analysis of COSTBC, we
consider a particular choice of $E_0 = \frac{E}{2}$ and $E_1 =
\frac{E}{2M_1}$ (half the total power is used by the transmitter and
half is equally distributed among all the relays).
In the following, we show that with this power allocation,
the diversity gain of COSTBC
is equal to the upper bound (Theorem \ref{upbound}) and thus
we do not lose any diversity gain by restricting the calculation
to this particular power allocation. Moreover, this
power allocation also satisfies the power constraint and therefore
provides us with a upper bound on the PEP.
Using this power allocation and the value
of $\gamma = E_0M_0L + L$,
\[\frac{E_0E_1M_1}{4M_0L\left(\gamma+E_1M_1^2\lambda_0\right)} \ge
\frac{E}{8M_1M_0L(\frac{L(M_0+1)}{M_1} + \lambda_0))}\]
for $E>1$, which implies
\begin{equation}
\label{pep4}
P\left({\bS}_{1l} \rightarrow {\bS}_{1m}|\lambda_{\bG}=\lambda_0\right) \le
{\bbE}_{\{{\hat \bH},\bG\}}
e^{-\frac{E}{8M_1M_0L(\mu + \lambda_0)}
\tr\left(\left({\bS}_{1l} - {\bS}_{1m}\right)\left({\bS}_{1l} - {\bS}_{1m}\right)^{*}\Phi\Phi^{*}\right)}
\end{equation}
where $\mu=\frac{L(M_0+1)}{M_1}$.

Recall that $\Phi = {\hat \bH}^{\frac{1}{2}}{\bG}$.
Let $\phi_j$ be the $j^{th}$ column of $\Phi$, then
\begin{eqnarray*}
\tr\left(\left({\bS}_{1l} - {\bS}_{1m}\right)\left({\bS}_{1l} - {\bS}_{1m}\right)^{*}\Phi\Phi^{*}\right) &=&
\tr\left(\Phi^{*}\left({\bS}_{1l} - {\bS}_{1m}\right)^{*}\left({\bS}_{1m} - {\bS}_{1l}\right)\Phi\right)\\
&=& \sum_{j=1}^{M_2}\phi_j^{*}
\left({\bS}_{1l} - {\bS}_{1m}\right)^{*}\left({\bS}_{1l} - {\bS}_{1m}\right)\phi_j.
\end{eqnarray*}
Thus, from (\ref{pep4})
\begin{eqnarray*}
P\left({\bS}_{1l} \rightarrow {\bS}_{1m}|\lambda_{\bG}=\lambda_0\right) &\le
 &
{\bbE}_{\{{\hat \bH}, \bG\}}
e^{-\frac{E}{8M_1M_0L(\mu + \lambda_0)}\sum_{j=1}^{M_2}\phi_j^{*}
\left({\bS}_{1l} - {\bS}_{1m}\right)^{*}\left({\bS}_{1l} - {\bS}_{1m}\right)\phi_j}\\
&\le & {\bbE}_{\{{\hat \bH}, \bG\}}
e^{-\frac{E}{8M_1M_0L(\mu + \lambda_0)}\sum_{j=1}^{M_2}\bg_j^{*}{\hat \bH}^{\frac{1}{2}*}
\left({\bS}_{1l} - {\bS}_{1m}\right)^{*}\left({\bS}_{1l} - {\bS}_{1m}\right){\hat \bH}^{\frac{1}{2}}\bg_j} \\
&\le & {\bbE}_{\{{\hat \bH}, \bG\}}
\prod_{j=1}^{M_2}
e^{-\frac{E}{8M_1M_0L(\mu + \lambda_0)}
\bg_j^{*}{\hat \bH}^{\frac{1}{2}*}
\left({\bS}_{1l} - {\bS}_{1m}\right)^{*}\left({\bS}_{1l} - {\bS}_{1m}\right){\hat \bH}^{\frac{1}{2}}\bg_j}
\end{eqnarray*}
where $\bg_j$ is the $j^{th}$ column of $\bG$. Since $\bg_j$ is a $M_1$ dimensional
Gaussian vector $\forall j=1,2,\ldots,M_2$, it follows that
\[
P\left({\bS}_{1l} \rightarrow {\bS}_{1m}|\lambda_{\bG}=\lambda_0\right) \le
{\bbE}_{\{{\hat \bH}\}}
\left[
\det\left(
{\bf I}_{M_1}+ \frac{E}{8M_1M_0L(\mu + \lambda_0)}{\hat \bH}^{\frac{1}{2}*}
\Delta \bS_{1lm}{\hat \bH}^{\frac{1}{2}}
\right)
\right]^{-M_2}
\] where $\Delta \bS_{1lm} \bydef \left({\bS}_{1l} - {\bS}_{1m}\right)^{*}\left({\bS}_{1l} - {\bS}_{1m}\right)$.
Since ${\bS}_1$ is an OSTBC the minimum singular value
$\sigma_{min}$ of $\Delta \bS_{1lm}$ is $>0$, which implies
\begin{equation}
\label{exph}
P\left({\bS}_{1l} \rightarrow {\bS}_{1m}|\lambda_{\bG}=\lambda_0\right) \le
{\bbE}_{\{{\hat \bH}\}}
\left[
\det\left(
{\bf I}_{M_1}+ \frac{E\sigma_{min}}{8M_1M_0L(\mu + \lambda_0)}
{\hat \bH}\right)
\right]^{-M_2}.
\end{equation}
Now we are left with computing the expectation in (\ref{exph})
with respect to ${\hat \bH}$.
Towards that end, recall that ${\hat \bH}$ is a diagonal matrix with
each entry $\sum_{m=1}^{M_0}|h_{mk}|^2$, which is gamma distributed
with probability density function PDF
$\frac{1}{(M_0-1)!}{x}^{M_0-1}e^{-{x}}$.

Therefore,
\[
P\left({\bS}_{1l} \rightarrow {\bS}_{1m}|\lambda_{\bG}=\lambda_0\right) \le
\frac{1}{(M_0-1)!}
\left[\int_0^{\infty}\left(1+ \frac{E\sigma_{min}}{8M_1M_0L(\mu + \lambda_0)}x\right)^{-M_2}x^{M_0-1}e^{-x}
dx\right]^{M_1}.
\]

Using an integration result from Theorem 3 \cite{Jing2004d}, it follows that
\begin{eqnarray}
\label{divcond}
\nonumber
P\left({\bS}_{1l} \rightarrow {\bS}_{1m}|\lambda_{\bG}=\lambda_0\right) & \le&\frac{1}{(M_0-1)!}
\left(\frac{8M_1M_0L(\mu + \lambda_0)}{\sigma_{min}}\right)^{\min{\{M_0,M_2\}}M_1
}\\
&&\times\left\{\begin{array}{cc}
\frac{(2^{M_0}-1)}{M_2-M_0}E^{-M_0M_1} & \ \text{if} \ \ M_2 \ge M_0 \\
\left(\frac{{\log E}^{\frac{1}{M_0}}}{E} \right)^{-M_0M_1}
& \ \text{if} \ \ M_2 = M_0 \\
(M_0-M_2-1)^{M_1}E^{-M_2M_1} & \ \text{if} \ \ M_2 \le M_0
\end{array}
\right.
\end{eqnarray}
for large transmit power $E$ and considering only the highest order terms of
$E$. Recall that
\begin{equation}
\label{pepcond}
P\left({\bS}_{1l} \rightarrow {\bS}_{1m}\right) = \bbE_{\{\lambda_{\bG}\}}P\left({\bS}_{1l} \rightarrow {\bS}_{1m}|
\lambda_{\bG}=\lambda_0\right).
\end{equation}
To evaluate this expectation we need to find the PDF of $\lambda_{\bG}$.
It turns out that explicitly finding the PDF of $\lambda_{\bG}$ is quite
difficult. To simplify the problem we use an upper bound on the PDF 
of $\lambda_{\bG}$ which is summarized in the next lemma.

\begin{lemma}
\label{wishartpdf}
For $M_1 \ge M_2$, the PDF of the maximum eigenvalue $\lambda_{\bG}$ of $\frac{1}{M_1}\bG^*\bG$ can be
upper bounded as
\[f_{\lambda_{\bG}}(\lambda_0) \le k_1\lambda_0^{M_1M_2-1}e^{-M_1\lambda_0}\ \]
where
\[k_1 = \frac{2^{M_2-1}M_1^{M_1M_2}}{\prod_{j=1}^{M_2}\Gamma(M_1-M_2+1)\Gamma(j)
\prod_{j=1}^{M_2}(M_1-M_2+2j-1)(M_1-M_2+2j(M_1-M_2+2j+1))}.\]
\end{lemma}
\vspace{0.2in}
\begin{proof}
Follows from Corollary $1$ \cite{Jing2004d}.
\end{proof}
\begin{rem}
From here on we evaluate the expectation in PEP upper bound for the case of $M_1 \ge M_2$ only.
For the other case, the analysis follows similarly, since
the PDF of $\lambda_{max}\left(\frac{\bG^*\bG}{M_1}\right)$, for $M_2<M_1$, can be
obtained from  Lemma \ref{wishartpdf} by switching the roles of
$M_1$ and $M_2$ and using the fact that $\lambda_{max}\left(\frac{\bG^*\bG}{M_1}\right)
= \lambda_{max}\left(\frac{\bG\bG^*}{M_1}\right)$.
\end{rem}
From (\ref{pepcond}),
\begin{eqnarray*}
P\left({\bS}_{1l} \rightarrow {\bS}_{1m} \right) &= &\int_{0}^{\infty} P\left({\bS}_{1l} \rightarrow {\bS}_{1m} |
\lambda_{\bG} = \lambda_0 \right)f_{\lambda_{\bG}}(\lambda_0)d\lambda_0.
\end{eqnarray*}
Using Lemma \ref{wishartpdf} and (\ref{divcond}),
\begin{eqnarray}
\label{pepquasi}
\nonumber
P\left({\bS}_{1l} \rightarrow {\bS}_{1m} \right)
& \le & 
\int_{0}^{\infty}\left(\frac{8M_1M_0L(\mu+\lambda_0)}{\sigma_{min}}\right)^{\min{\{M_0,M_2\}}M_1
} \\
&&\lambda_0^{M_1M_2-1}e^{-M_1\lambda_0}d\lambda_0\times \left\{\begin{array}{cc}
\frac{(2^{M_0}-1)}{M_2-M_0}E^{-M_0M_1} &   \ \text{if} \ \ M_2 \ge M_0 \\
\left(\frac{{\log E}^{\frac{1}{M_0}}}{E} \right)^{-M_0M_1}
& \ \text{if} \ \ M_2 = M_0 \\
(M_0-M_2-1)^{M_1}E^{-M_2M_1} &  \ \text{if} \ \ M_2 \le M_0
\end{array}
\right. .
\end{eqnarray}
Moreover, defining
\begin{eqnarray*}
k_2 &\bydef& \int_{0}^{\infty}  (\mu+\lambda_0)^{\min{\{M_0,M_2\}}M_1}\lambda_0^{M_1M_2-1}e^{-M_1\lambda_0} \  d\lambda_0\\
& =& \frac{c^i\sum_{i=0}^{\min\{M_0,M_2\}M_1}{\min\{M_0,M_2\}M_1\choose i}(\min\{M_0,M_2\}M_1+M_0-(i+1))!}{M_1^{-\left(\min\{M_0,M_2\}M_1+M_2M_1-i\right)}}
\end{eqnarray*}
the upper bound on PEP (\ref{pepquasi}) simplifies to
\begin{equation}
\label{diversity2hop}
P\left({\bS}_{1l} \rightarrow {\bS}_{1m}\right)
\le
k_3\times\left\{\begin{array}{cc}
\frac{(2^{M_0}-1)}{M_2-M_0}E^{-M_0M_1} &  \ \text{if} \ \ M_2 \ge M_0 \\
\left(\frac{{\log E}^{\frac{1}{M_0}}}{E} \right)^{-M_0M_1}
&  \ \text{if} \ \ M_2 = M_0 \\
(M_0-M_2-1)^{M_1}E^{-M_2M_1} &  \ \text{if} \ \ M_2 \le M_0
\end{array}
\right.,
\end{equation}
where $k_3 =\frac{k_1k_2}{\left((M_0-1)!\right)}
\left(\frac{8M_1M_0L}{\sigma_{min}}\right)^{\min{\{M_0,M_2\}}M_1
}$.
By the definition of diversity gain, from (\ref {diversity2hop}) it is
clear that diversity gain of COSTBC is $\min{\{M_0,M_2\}}M_1$, which equals
the upper bound from Theorem \ref{upbound}.
\end{proof}

Next we provide an alternate and simpler proof of Theorem \ref{2hopdiv}.
The outage probability formulation \cite{Zheng2003} and the single symbol
decodable property of COSTBCs is used to derive
 this proof. The purpose of this alternative proof is to highlight the fact
that the single symbol decodable property of COSTBCs not only minimizes the
decoding complexity but also improves analytical tractability.

\begin{proof} (Theorem \ref{2hopdiv})
The outage probability $P_{out}(R)$ is defined as
\[P_{out}(R) \bydef P\left(I(\bs;\br)\le R\right),\]
where $\bs$ is the input and $\br$ is the output of the channel
and $I(\bs;\br)$ is the mutual information between $\bs$ and $\br$ \cite{Cover2004}.

Let $\SNR \bydef \frac{E}{\sigma^2}$.
Following \cite{Zheng2003}, let ${\cal C}(\SNR)$ be a family of codes one for
each $\SNR$. We define $r$ as the spatial multiplexing gain of
${\cal C}(\SNR)$ if the data rate $R(\SNR)$ scales as $r$ with respect to
$\log \SNR$, i.e.
\[\lim_{\SNR\rightarrow \infty}\frac{R(\SNR)}{\log \SNR} =r\]
and $d$ as the rate of fall of probability of error $P_e$ of ${\cal
C}(\SNR)$ with respect to $\SNR$, i.e.
\[P_{e}(\SNR) \expeq \SNR^{-d}.\]
Let $d_{out}(r)$ be the $\SNR$ exponent of $P_{out}$ with rate of transmission
$R$ scaling as $r\log \SNR$, i.e.
\[\log P_{out}(r\log \SNR)\expeq \SNR^{-d_{out}(r)},\]
then it is shown in \cite{Zheng2003} that
\[P_{e}(\SNR) \expeq P_{out}(r\log \SNR) \expeq \SNR^{-d_{out}(r)}.\]
Thus, to compute the diversity gain of any coding scheme it is sufficient
to compute $d_{out}(r)$. In the following we compute $d_{out}(r)$ for the
COSTBC with a $2$-hop network.

For the $2$-hop network, using the single symbol decodable property of COSTBCs
(Appendix \ref{ssdcostbc}), the received signal can be separated in terms of
the individual constituent symbols of the OSTBC transmitted by the source.
Therefore, the received signal can be written as
\begin{equation}
\label{indhypfullchan}
r_l = \sqrt{\theta E}\sum_{j=1}^{M_2}\sum_{k=1}^{M_1}|g_{kj}|^2\left(\sum_{m=1}^{M_0}|h_{mk}|^2\right) s_l + z_l
\end{equation}
where $\theta$ is the normalization constant so as to ensure the total power
constraint of $E$ in the network,
$s_l$ is the $l^{th}, \ \ l=1,2,\ldots,L$ symbol transmitted from the
source and $z_l$ is the additive white Gaussian noise (AWGN) with variance $\sigma^{2}$.
Let $\SNR \bydef \frac{\theta E}{\sigma^2}$, then
\begin{eqnarray*}
P_{out}(r\log\SNR) & = & P\left(1+\SNR\sum_{j=1}^{M_2}\sum_{k=1}^{M_1}|g_{kj}|^2\left(\sum_{m=1}^{M_0}|h_{mk}|^2\right) \le r\log \SNR\right) \\
 & \expl & P\left(\sum_{k=1}^{M_1}\sum_{j=1}^{\min\{M_0,M_2\}}|g_{kj}|^2|h_{jk}|^2 \le \SNR^{-(1-r)}\right) \\
&\le&P\left(\max_{\{j=1,\ldots,\min\{M_0,M_2\}, \ k=1,\ldots,M_1\}}|g_{kj}|^2|h_{jk}|^2 \le \SNR^{-(1-r)}\right).
\end{eqnarray*}
Since $|g_{kj}|^2|h_{jk}|^2$ are i.i.d. for
$j=1,\ldots,\min\{M_0,M_2\}, \ k=1,\ldots,M_1$ and
the total number of terms are $\min\{M_0M_1,\ M_1M_2\}$,
\begin{eqnarray*}
P_{out}(r\log\SNR) &\expeq& P\left(|g_{11}|^2|h_{11}|^2 \le \SNR^{-(1-r)}\right)^{\min\{M_0M_1,\ M_1M_2\}}.
\end{eqnarray*}
Note that $P\left(|g_{11}|^2|h_{11}|^2 \le \SNR^{-(1-r)}\right)$ is the outage probability of
a single input single output system which can be computed easily using
\cite{Zheng2003} and is given by
\[P\left(|g_{11}|^2|h_{11}|^2 \le \SNR^{-(1-r)}\right) \expeq \SNR^{-(1-r)}, \ r\le 1.\]
Thus,
\[P_{out}(r\log\SNR) \expeq \SNR^{-\min\{M_0M_1,\ M_1M_2\}(1-r)}, \ r\le 1,\]
and we have shown that $d_{out}(r) = \min\{M_0M_1,\ M_1M_2\}(1-r), \ r\le 1$,
from which it follows that the diversity gain of COSTBC is
$d_{out}(0) = \min\{M_0M_1,\ M_1M_2\}$ as required.
\end{proof}
{\it Discussion:}
In this section we derived an upper bound on the PEP of COSTBCs for a $2$-hop network from which we lower bounded
the diversity gain of COSTBCs for a $2$-hop network. We showed that the lower
bound on the diversity gain of COSTBCs equals the upper bound from Theorem \ref{upbound} and thus concluded that COSTBCs achieve the maximum diversity gain in a $2$-hop
network.

We presented two different proofs that show the optimality of COSTBCs in
the sense of achieving the maximum diversity gain in $2$-hop network. In the first proof
we directly worked with the PEP using maximum likelihood
detection while in the second proof we used the outage probability formulation
\cite{Zheng2003}. The purpose of giving two proofs is to highlight the different
ideas one can use to upper bound the PEP of multi-antenna multi-hop
communication systems for possible extensions to more complex channels.

The main difficulty in upper bounding the PEP of COSTBCs was due to the fact that
the covariance matrix $\bR_{\bW}$ of noise received at the
destination is not a diagonal matrix.
In the first proof we simplified the problem by upper bounding
the maximum eigenvalue of $\bR_{\bW}$ by the eigenvalues
of $\frac{\bG\bG^{*}}{M}$ and
then used standard techniques to upper bound the PEP.
In the second proof we used the outage probability formulation \cite{Zheng2003}
to lower bound the diversity gain of COSTBCs for $2$-hop network.
To upper bound the outage probability, we used the single symbol decodable property of COSTBCs and
showed that the exponent of the outage probability with COSTBCs
is $\min\{M_0M_1, M_1M_2\}$ times the exponent of the outage probability of
 SISO system whose diversity gain is $1$. Thus we concluded that the
diversity gain of COSTBCs is $\min\{M_0M_1, M_1M_2\}$.

\section{Diversity Gain Analysis of COSTBC for Multi-Hop Case}
\label{sec:multihop}
In this section we show that COSTBCs achieve
the maximum diversity for a $N$-hop network where $N\ge 2$.
Recall that with COSTBC the source and each relay stage use an
OSTBC to communicate with
the following relay stage. With CSI available at each relay,
in Appendix \ref{ssdcostbc} we show that COSTBCs have the single
symbol decodability property
similar to OSTBC. Thus, with the COSTBCs
each of the constituent symbols of the OSTBC transmitted by the source
can be decoded independently of all the other symbols at any relay of any relay
stage or at the
destination without any loss in performance compared to joint decoding.
We use this property to show that
the COSTBCs achieve the upper bound on the diversity gain of an $N$-hop network
given by Theorem \ref{upbound}.

%

\begin{thm}
With COSTBCs, a
diversity gain of $\min\{M_nM_{n+1}\} \ \ n=0,1,\ldots,N-1$ is achievable for
a $N$-hop
network.
\end{thm}

\begin{proof} We use induction to prove the Theorem.
From Section \ref{sec:2hop} the result is true for a $2$-hop network, and
hence we can start the induction.
Now assume that the result is true for a $k$-hop network $(k\ge 2) $
and we will prove that it is true for a $k+1$-hop network.

For a $k$-hop network using the single symbol decodable property of COSTBCs
as shown in Appendix \ref{ssdcostbc}, at the destination the
received signal can be separated in terms of the individual
constituent symbols of the OSTBC transmitted by the source. Thus the
received signal can be written as
\begin{equation}
\label{indhypfullchan}
r_{\ell} = \sqrt{\theta E}\sum_{i=1}^{M_{k}}c_{i} s_{\ell} + z_{\ell},
\end{equation}
where $\theta$ is the normalization constant so as to ensure the total power
constraint of $E$ in the network,
$s_{\ell}$ is the ${\ell}^{th}, \ \ \ell=1,2,\ldots,L$
symbol transmitted from the source, $c_{i}$ is the channel gain experienced
by $s_{\ell}$ at the $i^{th}$ antenna of the destination, and $z_l$ is the additive white Gaussian noise (AWGN) with variance $\sigma_{k}^2$.

Now we extend the $k$-hop network to a $k+1$-hop network by
assuming that the actual destination to be one more hop away and
using the destination of the $k$-hop case as
the $k^{th}$ relay stage with $M_k$ relays by separating the $M_{k}$ antennas into $M_{k}$
relays with single antenna each.
Again using the single symbol decodable property of COSTBCs for the $k+1$-hop network,
as shown in the Appendix \ref{ssdcostbc}, the received
signal at the
destination can be separated in terms of individual constituent symbols of the OSTBC transmitted
by the source, which is given by
\[y_{\ell} = \sqrt{\kappa E}\sum_{i=1}^{M_{k}}c_{i}\left(\sum_{j=1}^{M_{k+1}}|g_{ij}|^2\right)s_{\ell} + n_{\ell}, \ \ \ell = 1,\ldots, L\]
where $\kappa$ is a constant to
ensure the power constraint of $E$ in the $k+1$-hop network,
$g_{ij}$ is the channel between the $i^{th}$ relay of relay stage $k$ and
the $j^{th}$ antenna of the destination and $n_l$ is the AWGN with variance
$\sigma^2_{k+1}$.

Defining $q\bydef \sum_{i=1}^{M_k}q_i$ and $q_i \bydef c_{i}\left(\sum_{j=1}^{M_{k+1}}|g_{ij}|^2\right)$, we can write

\begin{equation}
\label{indfullchan}
y_{\ell} = \sqrt{\kappa E}qs_{\ell} + n_{\ell}
\end{equation}
and
\begin{equation}
\label{indpartchan}
y_{\ell i} = \sqrt{\kappa E}q_is_{\ell} + n_{\ell i}
\end{equation}
for each $\ell=1,\ldots, L$, where $y_{\ell} = \sum_{i=1}^{M_k}y_{\ell i}$ and $n_{\ell i} = n_{\ell}/M_{k}$.

Recall from induction hypothesis that the diversity gain of COSTBCs with
channel $c_{i}, \ \forall
i$ (\ref{indhypfullchan}) is
$\alpha \bydef \min\left\{\min\left\{M_nM_{n+1}\right\},\ M_{k-1}\right\}, \ \
n=0,1,\ldots,k-2$,
by restricting the destination of the $k$-hop network to have only single
antenna, and with channel
$\sum_{i=1}^{M_{k}}c_{i}$ is $\min\left\{M_nM_{n+1}\right\},
\ n=0,1,\ldots,k-1$, respectively. Thus, if the diversity
gain of  COSTBCs with channel $q_i$ (\ref{indpartchan}) is
$\min\left\{\min\left\{M_nM_{n+1}\right\},\ M_{k-1}, M_{k+1}\right\}$ $
n=0,1,\ldots,k-2$,
then, since $\sum_{j=1}^{M_{k+1}}|g_{ij}|^2$ are independent $\forall \ i$,
it follows that the diversity gain of COSTBCs with channel
$\sum_{i=1}^{M_{k}}q_i$ is  $\min\left\{M_nM_{n+1}\right\},
\ n=0,1,\ldots,k$.
Next, we show that the diversity gain of COSTBCs with
channel $q_i$ is $\min\left\{\min\left\{M_nM_{n+1}\right\},\ M_{k-1}, M_{k+1}\right\}, \ \
n=0,1,\ldots,k-2$.

To compute the diversity gain of COSTBCs with channel $q_i$ (\ref{indpartchan}),
we use the outage probability formulation \cite{Zheng2003} as follows.
Let $\sigma^2$ be the variance of $n_{\ell i}$ (\ref{indpartchan}), 
$\sigma^2 = \frac{\sigma^2_{k+1}}{M_{k}^2}$, and as before $\SNR \bydef 
\frac{\kappa E}{\sigma^2}$, then the outage probability of (\ref{indpartchan}) is
\begin{eqnarray*}
P_{out}(r\log \SNR) & \bydef & P \left(\log\left(1 + 
\SNR c_i\sum_{j=1}^{M_{k+1}}|g_{ij}|^2\right) \le r \log \SNR) \right).
\end{eqnarray*}
{\small
\begin{eqnarray}
\nonumber
P_{out}(r\log \SNR) & \expeq & P \left(c_i\sum_{j=1}^{M_{k+1}}|g_{ij}|^2\le \SNR^{-(1-r)} \right) \\\nonumber
& = & P \left(\sum_{j=1}^{M_{k+1}}|g_{ij}|^2\le \SNR^{-(1-r)} \right)
P \left(c_i\sum_{j=1}^{M_{k+1}}|g_{ij}|^2\le \SNR^{-(1-r)} |
\sum_{j=1}^{M_{k+1}}|g_{ij}|^2\le \SNR^{-(1-r)}\right) \\ \nonumber
&& + \ P\left(\sum_{j=1}^{M_{k+1}}|g_{ij}|^2 > \SNR^{-(1-r)} \right)
P\left(c_i\sum_{j=1}^{M_{k+1}}|g_{ij}|^2\le \SNR^{-(1-r)} | \sum_{j=1}^{M_{k+1}}|g_{ij}|^2 > \SNR^{-(1-r)}\right) \\ \nonumber
& \le & P \left(\sum_{j=1}^{M_{k+1}}|g_{ij}|^2\le \SNR^{-(1-r)} \right) +
P\left(c_i\sum_{j=1}^{M_{k+1}}|g_{ij}|^2\le \SNR^{-(1-r)} | \sum_{j=1}^{M_{k+1}}|g_{ij}|^2 > \SNR^{-(1-r)}\right).
\end{eqnarray}
Let $Z\bydef \sum_{j=1}^{M_{k+1}}|g_{ij}|^2$. Then
\begin{eqnarray*}
P_{out}(r\log \SNR)
& \expl & P \left(Z\le \SNR^{-(1-r)} \right) +
\int_{\SNR^{-(1-r)}}^{\infty} \int_{0}^{\SNR^{-(1-r)}/z} f_{c_i}(y) dy f_Z(z)dz.
\end{eqnarray*}}
By induction hypothesis, the diversity gain of COSTBCs with $c_i$ is
$\alpha$, i.e.,
\[P\left(c_i\le \frac{\SNR^{-(1-r)}}{z}\right) = \int_{0}^{\SNR^{-(1-r)}/z} f_{c_i}(y) dy \le k_4
\left(\frac{\SNR^{-(1-r)}}{z}\right)^{\alpha}\]
where $k_4$ is a constant. Thus,
\begin{eqnarray}
\label{poutmultihop}
P_{out}(r\log \SNR)
& \le & P \left(Z\le \SNR^{-(1-r)} \right) +
\int_{\SNR^{-(1-r)}}^{\infty} k_4 \SNR^{-\alpha(1-r)}\left(\frac{1}{z}\right)
^{\alpha}
  f_Z(z)dz.
\end{eqnarray}
Since $Z$ is a gamma distributed random variable with PDF
$\frac{e^{-z}z^{M_{k+1}-1}}{M_{k+1}-1!}$, the first term in $P_{out}(r\log \SNR)$
expression can be found in \cite{Zheng2003} and is given by
\[P \left(Z\le \SNR^{-(1-r)} \right) \expeq \SNR^{-M_{k+1}(1-r)}.\]
Now we are left with computing the second term which can be done as follows.
\begin{eqnarray*}
\int_{\SNR^{-(1-r)}}^{\infty}k_4 \SNR^{-\alpha(1-r)}\left(\frac{1}{z}\right)^{\alpha}f_Z(z)dz &
= & k_4 \SNR^{-\alpha(1-r)} \int_{\SNR^{-(1-r)}}^{\infty}
z^{-\alpha}\frac{e^{-z}z^{M_{k+1}-1}}{M_{k+1}-1!} dz\\
&=& \frac{k_4}{M_{k+1}-1!} \SNR^{-\alpha(1-r)} c_5,
\end{eqnarray*}
where
\[c_5 \le
\left\{\begin{array}{cc} M_{k+1}-\alpha -1! &  \ \text{if} \ \ \alpha < M_{k+1} \\
(-1)^{\alpha-M_{k+1}+1}
\frac{Ei(-\SNR^{-(1-r)})}{\alpha-M_{k+1}}+
\sum_{k=0}^{\alpha-M_{k+1}-1}
\frac{(-1)^k\exp^{(-\SNR^{-(1-r)})}\SNR^{-k(1-r)}}
{(\alpha-M_{k+1})(\alpha-M_{k+1}-1)\ldots(\alpha-M_{k+1}-k)}
&  \ \text{if} \ \ \alpha \ge M_{k+1}
\end{array}
\right.\ \]
from \cite{Gradshteyn1994}.
Thus, from (\ref{poutmultihop}) it follows that
\[P_{out}(r \log \SNR) \expeq \SNR^{-M_{k+1}(1-r)} + \SNR^{-\alpha(1-r)}.\]
which implies that
\[P_{out}(r \log \SNR) \expeq \SNR^{-\min\{M_{k+1},\alpha\}(1-r)}.\]
Using the definition of diversity gain, it follows that the diversity
gain of COSTBCs with channel $q_i$ is equal to $\min\{\alpha, M_{k+1}\}$,
which implies that the diversity gain of COSTBCs with channel $q$
(\ref{indfullchan}) is $\min\{\alpha M_k, M_kM_{k+1}\}$.
Note that the upper bound on the
diversity gain (Theorem \ref{upbound}) is also $\min\{\alpha M_k, M_kM_{k+1}\}$
and we conclude that the COSTBCs achieve the
maximum diversity gain in a $N$-hop network.

\end{proof}
{\it Discussion:} In this section we showed that COSTBCs achieve
a diversity gain of $\min\left\{M_nM_{n+1}\right\}$
$n=0,1,\ldots,N-1$
in an $N$-hop network which equals the upper bound obtained in
Theorem \ref{upbound} for arbitrary integer $N$.
Thus we showed that the COSTBCs are optimal in terms of achieving the maximum
diversity gain of $N$-hop network.

To obtain this result we used
the single symbol decodable property of COSTBCs and mathematical induction.
Using the single symbol decodable property we were able to decouple the different constituent
symbols of the OSTBC transmitted by the source, at the destination
which made the diversity gain analysis easy.

\section{Code Design}
\label{sec:code}
In this section, we explicitly construct COSTBCs that
achieve maximum diversity gain in $N$-hop networks.
We present examples of COSTBCs for $N=2$, $M_0=M_1=2$ using the Alamouti
code \cite{Alamouti1998}, $N=2$, $M_0=M_1=4$ using the rate-$3/4$ $4$ antenna OSTBC
\cite{Tarokh1999} and $N=2$, $M_0=M_1=4$ using
the rate-$3/4$ $4$ antenna OSTBC and the Alamouti code.
\begin{exm}(Cascaded Alamouti Code)
\label{codecascala}
We consider $N=2$, $M_0=M_1=2$ case and let $\bS_0$ be the Alamouti code
given by:$\left[\begin{array}{cc}
s_1 & s_2 \\
-s_2^{*} & s_1^{*}
\end{array}\right]$
where $s_1$ and $s_2$ are constituent symbols of the Alamouti code.
The $2\times 1$ received signal at relay $m$ is
\[\left[\begin{array}{c}
r_{1m} \\
r_{2m} \end{array}\right] = \sqrt{E_0}\left[\begin{array}{cc}
s_1 & s_2 \\
-s_2^{*} & s_1^{*}
\end{array}\right]\left[\begin{array}{c}
h_{1m} \\
h_{2m} \end{array}\right] + \left[\begin{array}{c}
n_{1m} \\
n_{2}m \end{array}\right]\]
for $m=1,2$.
Transforming this in the usual way
\[\left[\begin{array}{c}
r_{1m} \\
-r_{2m}^* \end{array}\right] = \sqrt{E_0}
\underbrace{\left[\begin{array}{cc}
h_{1m} & h_{2m} \\
-h_{2m}^{*} & h_{1m}^{*}
\end{array}\right]}_{\tilde\bH_m}\left[\begin{array}{c}
s_{1} \\
s_{2} \end{array}\right] + \left[\begin{array}{c}
n_{1m} \\
-n_{2m}^{*} \end{array}\right]\]
for $m=1,2$.
We define ${\tilde h}_m \bydef |h_{1m}|^2+ |h_{2m}|^2$,
$\eta_{1m} \bydef (n_{1m}h_{1m}^{*} + n_{2m}^{*}h_{2m})$, and
$\eta_{2m} \bydef (n_{1m}h_{2m}^{*} - n_{2m}^{*}h_{1m})$.
Pre-multiplying by ${\tilde \bH_m}^{*}$,
\begin{eqnarray*}
\label{relaytxala}
\left[\begin{array}{c}
\hat{r}_{1m} \\
\hat{r}_{2m}^{*} \end{array}\right] \bydef
{\tilde \bH_m}^{*}\left[\begin{array}{c}
\hat{r}_{1m} \\
\hat{r}_{2m}^{*} \end{array}\right] &=&
\sqrt{E_0}\left[\begin{array}{c}
{\tilde h}_ms_{1}\\
{\tilde h}_ms_{2}
\end{array}\right] + \left[\begin{array}{c}
\eta_{1m}\\
\eta_{2m}\end{array}\right]
\end{eqnarray*}
for $m=1,2$.
Now using
\[\bA_1 = \left[\begin{array}{cc}
1 & 0 \\
0 & 1 \end{array}\right], \bB_1 = {\bf 0}_2, \ \ \bA_2 = {\bf 0}_2,
\bB_2 = \left[\begin{array}{cc}
0 & -1 \\
1 & 0 \end{array}\right]\]
the STBC $\bS_1$ formed by the two relays is of the form
$\left[\begin{array}{cc}
s_1 & -s_2^{*} \\
s_2 & s_1^{*}
\end{array}\right]$
which is an Alamouti code and hence an OSTBC as required.
Note that $\bA_i, \bB_i \ i=1,2$ satisfy the requirements of (\ref{hurrad}).
We call this the {\it cascaded Alamouti code}.

\end{exm}

\begin{exm}
In this example we consider the case $N=2$, $M_0=4$, $M_1=4$.
We choose $\bS_0$ to be
the rate-$3/4$ OSTBC for $4$ transmit antennas given by
\[\left[\begin{array}{cccc}       s_1     & s_2      & s_3           & 0\\
                                 -s_2^{*} & s_1^{*}  & 0             & s_3 \\
                                  s_3^{*} & 0        & -s_1^{*}      & s_2\\
                                   0      & s_3^{*}  & -s_2^{*}      & -s_1
\end{array}\right]\]
and use
\[\bA_1 = \left[\begin{array}{ccc} 1 & 0 & 0 \\
                   0 & 0 & 0\\
                   0 & 0 & 0\\
                   0 & 0 & 0\\
\end{array}\right],
\bA_2 = \left[\begin{array}{ccc}
                   0 & 1 & 0\\
                   0 & 0 &0\\
                   0 & 0 &0\\
                   0 & 0 &0\\
\end{array}\right],
 \ \
\bA_3 = \left[\begin{array}{ccc}
                   0 & 0 & 1 \\
                   0 & 0 & 0\\
                   0 & 0 & 0\\
                   0 & 0  & 0\\
\end{array}\right],
\ \
\bA_4 = \left[\begin{array}{ccc}
                   0 & 0 & 0 \\
                   0 & 0 & 1\\
                   0 & 1 & 0\\
                   -1 & 0 & 0\\
\end{array}\right]\]
and
\[\bB_1 = \left[\begin{array}{ccc} 0 & 0 & 0 \\
                   0 & -1 & 0\\
                   0 & 0 & 1\\
                   0 & 0 & 0\\
\end{array}\right],
\bB_2 = \left[\begin{array}{ccc}
                   0 & 0 & 0\\
                   1 & 0 &0\\
                   0 & 0 &0\\
                   0 & 0 &1\\
\end{array}\right],
 \ \
\bB_3 = \left[\begin{array}{ccc}
                   0 & 0 & 0 \\
                   0 & 0 & 0\\
                   -1 & 0 & 0\\
                   0 & -1  & 0\\
\end{array}\right],
\ \
\bB_4 = \left[\begin{array}{ccc}
                   0 & 0 & 0 \\
                   0 & 0 & 0\\
                   0 & 0 & 0\\
                   0 & 0 & 0\\
\end{array}\right].\]
It is easy to verify that $\tr{\left(\bA_i^{*}\bA_i + \bB_i^{*}\bB_i\right)} = 3$ and
$\bA_i^{*}\bB_i = -\bB_i^{*}\bA_i, \ \  i =1,2,3,4$ as required. The STBC
 $\bS_1$ using these $\bA_i,\bB_i\  i =1,2,3,4$ is
\[\left[\begin{array}{cccc}s_1 & s_2 & s_3 & 0\\
                              -s_2^{*} & s_1^{*} & 0 & s_3 \\
                               -s_3^{*} & 0 & s_1^{*} & s_2\\
                                0 & -s_3^{*}& s_2^{*} & -s_1
\end{array}\right]\]
which is a rate-$3/4$ OSTBC as described above.
\end{exm}
In both the previous examples we constructed COSTBC for $N=2$-hop case
by repeatedly using the same OSTBC at both the source and the relay stage.
Using a similar procedure, it is easy to see that when
$M_i=M_j \ \forall \ i,j=0,1,\ldots,N-1, \ i\ne j $
we can construct COSTBCs by using particular OSTBC for $M_0$ antennas
at the source and each relay stage, e.g. if ${\cal O}$ is an OSTBC
for $M_0$ antennas, then by using
$\bS_n= {\cal O}, \ n=0,1,\ldots, N-1$ we obtain a maximum diversity
gain achieving COSTBCs.
OSTBC constructions for different number of antennas can be found in
\cite{Tarokh1999}.

In the next example we construct COSTBC for $M_0=4$ and $M_1=2$ by
cascading the rate-$3/4$ $4$ antenna OSTBC with the Alamouti code.

\begin{exm}
Let $N=2,\ M_0=4$ and $M_1=2$.
We choose $\bS_0$ to be the rate-$3/4$ $4$ antenna
OSTBC. In this example each relay node accumulates $6$ constituent
symbols from $2$ blocks of $\bS_0$ and transmits them in $3$ blocks
of Alamouti code to the destination as follows.

Let $\bS_0^t$ be the transmitted rate-$3/4$ $4$ antenna
OSTBC at time $t, \ t=1,5,11,14,20,24,\ldots$ from the source and
$s_j^t, \ j=1,2,3$ be the $j^{th}$ constituent symbol of $\bS^t_0$, i.e.
\[\bS^t_0 = \left[\begin{array}{cccc}s^t_1 & s^t_2 & s^t_3 & 0\\
                              -s_2^{t*} & s_1^{t*} & 0 & s^t_3 \\
                               -s_3^{t*} & 0 & s_1^{t*} & s^t_2\\
                                0 & -s_3^{t*}& s_2^{t*} & -s^t_1
\end{array}\right].\]
Then the received signal at relay node $m, \ m=1,2$ at time $t=1,5,11,14,20,24,\ldots$ is
\[\br^t = \sqrt{E_0}\bS_0^t\left[\begin{array}{c}h_{1m} \\h_{2m} \\h_{3m} \\h_{4m}\end{array}\right] +  \left[\begin{array}{c}{n}_1^t \\{n}_2^t \\{n}_3^t \\{n}_4^t\end{array}\right].\]
Using CSI the received signal $\br^t$ can be transformed into ${\hat \br}^t$,
where
\[{\hat \br}^t \bydef \left[\begin{array}{c}{\hat r}_1^t \\{\hat r}_2^t \\{\hat r}_3^t\end{array}\right] = \sqrt{E_0} \left[\begin{array}{c}{\hat h}_ms^t_1 \\ {\hat h}_ms^t_2 \\{\hat h}_ms^t_3  \end{array}\right]+ \left[\begin{array}{c}{\hat n}_1^t \\{\hat n}_2^t \\{\hat n}_3^t\end{array}\right]\]
 and ${\hat h}_m = \sqrt{\sum_{i=1}^{M_0}|h_{im}|^2}$.
Then as described before, each relay accumulates $6$ constituent symbols
from $2$ consecutive transmissions of $\bS_0$ from the source, i.e.
from $\bS_0^1$ and $\bS_0^5$.
Then at time $t=9$, the relay $m, \ m=1,2$ transmits
\[\bA_m\left[\begin{array}{c}{\hat r}_1^1 \\{\hat r}_2^1\end{array}\right] + \bB_m\left[\begin{array}{c}{\hat r}_1^1 \\{\hat r}_2^1\end{array}\right]^{\dag}\] where
\[\bA_1 = \left[\begin{array}{cc} 1 & 0 \\ 0 & 1\end{array}\right], \bB_1 = {\bf 0}_2, \ \bA_2 = {\bf 0}_2, \ \bB_2 = \left[\begin{array}{cc} 0 & -1 \\ 1 & 0\end{array}\right].\] Thus at time $t=9$,
\[\bS_1 = \left[\begin{array}{cc}s^1_1 & -s_2^{1*} \\
                                 s_2^1 & s_1^{1*}\end{array}\right]\]
which is an Alamouti code transmitted from the relay stage $1$ to the destination.
Similarly, at time $t=11$, the relay $m$ transmits
\[\bA_m\left[\begin{array}{c}{\hat r}_3^{1} \\{\hat r}_1^5\end{array}\right] + \bB_m\left[\begin{array}{c}{\hat r}_3^{1} \\{\hat r}_1^5\end{array}\right]^{\dag},\] and at time $t=13$, the relay $m$ transmits
\[\bA_m\left[\begin{array}{c}{\hat r}_2^{5} \\{\hat r}_3^5\end{array}\right] + \bB_m\left[\begin{array}{c}{\hat r}_2^{5} \\{\hat r}_3^5\end{array}\right]^{\dag}.\]
These operations are repeated at the source and each relay stage
in subsequent time slots. Clearly, the relay stage transmits an Alamouti
code which is an OSTBC and hence we get an COSTBC construction for $M_0=4,
\ M_1=2$.
\end{exm}

Using a similar technique as illustrated in this example,
COSTBCs can be constructed for different number of source antenna and
relay node configurations by suitably adapting different OSTBCs.

\section{Simulation Results}
\label{sec:simulation}
In this section we provide some simulation
results to illustrate the bit error rates (BER) of COSTBCs for $2$ and $3$-hop networks.
In all the simulation plots, $E$ denotes the total power used by all nodes in
the network, i.e. $E_0 + \sum_{n=1}^{N-1}M_nE_n = E$ and the
additive noise at each relay and the destination is complex Gaussian with zero mean and unit
variance. By equal
power allocation between the source and each relay stage we mean $E_0=M_nE_n= \frac{E}{N}, \
 \forall n=1,\ldots, N-1$.

In Fig. \ref{2stage} we plot the bit error rates of a cascaded
Alamouti code and the comparable DSTBC from \cite{Jing2004d} with $4$
QAM modulation for $N=2$, $M_0=M_1=2$ and $M_2=1,2,3$ with equal
power allocation between the source and all the relays. It is easy
to see that both the cascaded Alamouti code and the DSTBC from
\cite{Jing2004d} achieves the maximum diversity gain of the $2$-hop
network, however, COSTBCs require $1 $ dB less power than the DSTBCs
from \cite{Jing2004d}, to achieve the same BER. The improved BER
performance of COSTBCs over DSTBCs from \cite{Jing2004d}, is due to
fact that with COSTBCs, each relay coherently combines the signal
received from the previous relay stage before forwarding it to the
next relay stage, while no such combining is done in
\cite{Jing2004d}.

To understand the effect of power allocation between the source and the relays
on the BER performance of cascaded Alamouti code, Fig. \ref{powercomp}
compares the BER performance of cascaded Alamouti code
for $N=2$, $M_0=M_1=2$ and $M_2=1$ with equal power allocation
and with power allocation of $ E_0= E/4$ at the source and $E_1 = 3E/8$ at
each relay.
It is clear that with unequal power allocation there is a gain of
around $1$ dB but no extra diversity gain. It turns out that it is difficult to
explicitly derive the best power allocation policy in terms of optimizing the
BER.

Next we plot the BER curves for $N=2$, $M_0=M_1=4$, and $N=2$, $M_0=4, \ M_1=2$
configurations in Figs. \ref{2stage4x4} and \ref{2stage4x2} with different
$M_2$ and using equal power allocation between the source and the relay stage.
For the $M_0=M_1=4$ case
we use the cascaded rate-$3/4$ $4$ antenna OSTBC
 and for the $M_0=4, \ M_1=2$ case we use a rate-$3/4$ $4$ antenna
OSTBC at the source and the Alamouti code across both the relays as discussed in
Section \ref{sec:code}.
In the $M_0=4, \ M_1=2$ case, both relays accumulate $6$ symbols from two
blocks of rate-$3/4$
$4$ antenna OSTBC and then relay these $6$ symbols in three blocks of
Alamouti code to the destination.
From Figs. \ref{2stage4x4} and \ref{2stage4x2} it is clear that both
these codes achieve maximum diversity gain for the respective network
configurations.

Finally, in Fig. \ref{3stage} we plot the bit error rates of
a cascaded Alamouti code with $N=3$-hop network where $M_0=M_1=M_2=2$ with
 $M_3=1,2,3$, and the cascaded Alamouti code is generated by
repeated use of the Alamouti code by each relay stage with equal power
allocation between the source and the relay stages. In this case also
it is clear that the cascaded Alamouti code achieves the maximum
diversity gain but there is a SNR loss compared to $N=2$ case, because
of the noise added by one extra relay stage.

From all the simulation plots, it is clear that COSTBCs
require large transmit power to obtain reasonable BER's with multi-hop wireless
networks.
This is a common phenomenon across all the maximum diversity gain achieving
DSTBC's for multi-hop wireless
networks that use AF \cite{Jing2004d, Jing2006a, Oggier2006k}.
With AF, the noise
received at each relay gets forwarded towards the destination and
limits the received SNR at the destination, however, without using AF it
is difficult to achieve maximum diversity gain in a multi-hop wireless network.

\begin{figure}
\centering
\includegraphics[height= 3in]{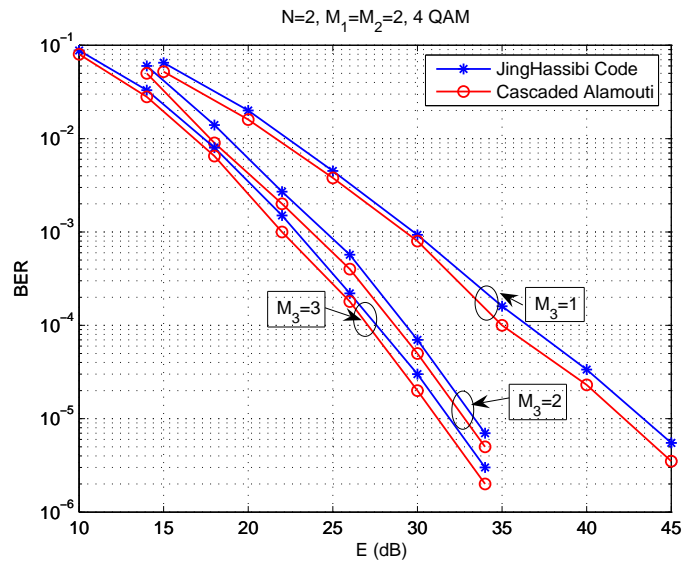}
\caption{BER comparison of Cascaded Alamouti code with JingHassibi
code for $N=2$-hop network} \label{2stage}
\end{figure}

\begin{figure}
\centering
\includegraphics[height= 3in]{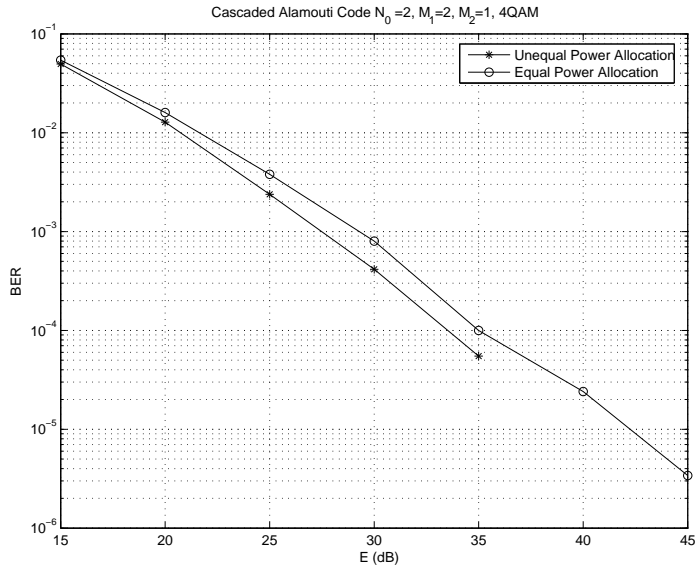}
\caption{Performance of cascaded Alamouti with varying power
allocation} \label{powercomp}
\end{figure}

\begin{figure}
\centering
\includegraphics[height= 3in]{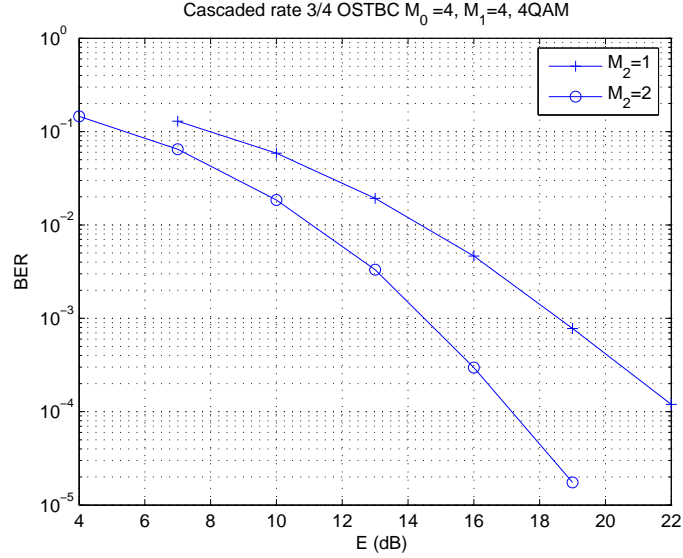}
\caption{Cascaded rate 3/4 4 antenna OSTBC for $M_0=M_1=4$}
\label{2stage4x4}
\end{figure}

\begin{figure}
\centering
\includegraphics[height= 3in]{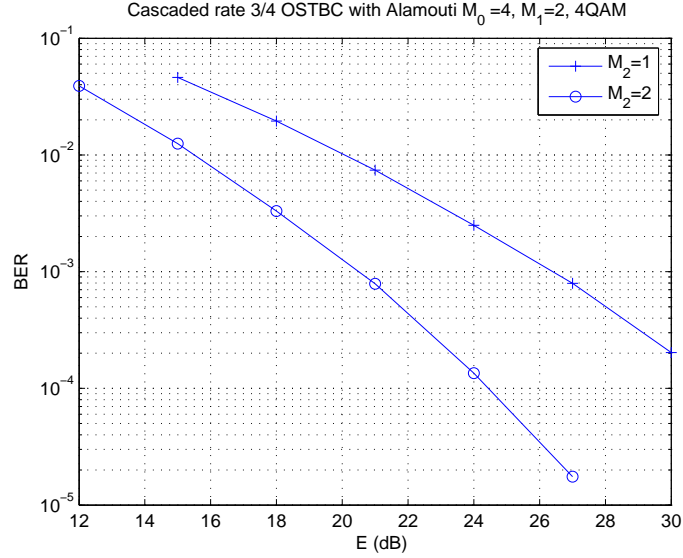}
\caption{Cascaded rate 3/4 4 antenna OSTBC with Alamouti Code for $M_0=M_1=4$ }
\label{2stage4x2}
\end{figure}

\begin{figure}
\centering
\includegraphics[height= 3in]{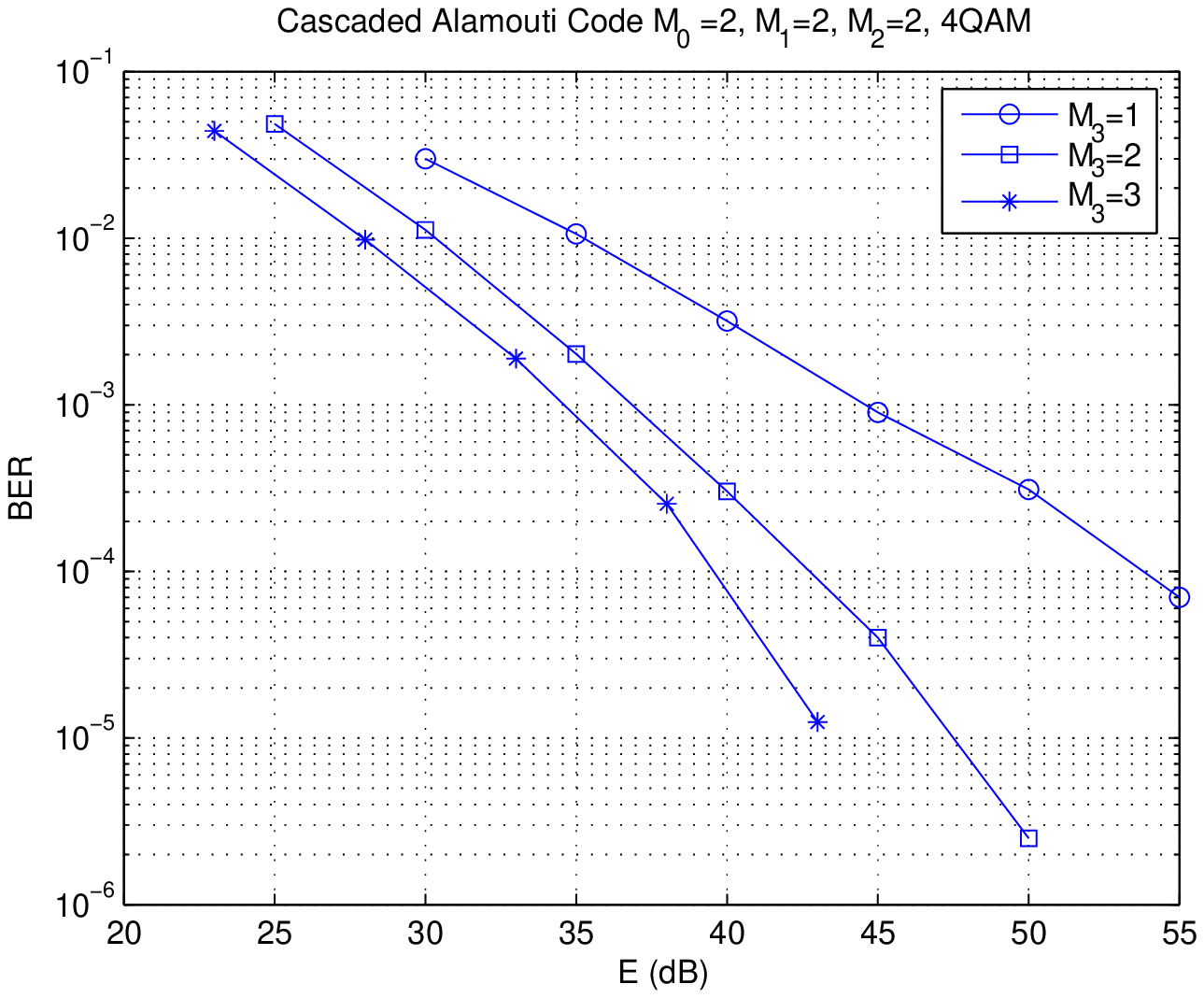}
\caption{Cascaded Alamouti Code for $N=3$-hop network}
\label{3stage}
\end{figure}

\section{Conclusion}
\label{sec:conc} In this paper we designed DSTBC's for
multi-hop wireless network and analyzed their diversity gain.
We assumed that receive CSI is
known at each relay and the destination. We proposed an AF strategy
called COSTBC to design DSTBC
using OSTBC to communicate between
adjacent relay stages when
CSI is available at each relay. We showed that the COSTBCs achieve the 
maximum diversity gain in a multi-hop wireless network. We also showed that
COSTBCs are single symbol decodable similar to OSTBC
and thus incur minimum decoding complexity.
We then gave an explicit construction of COSTBCs for various
numbers of source, destination,
and relay antennas that were shown to achieve maximum diversity gain with
minimal encoding complexity.
The only restriction that COSTBCs
impose is that the source and all the relay stages have to use an OSTBC.
 It is well known that high rate OSTBC do not exist,
therefore the COSTBCs have rate limitations. For future work it will
be interesting to see whether the OSTBC requirement can be relaxed
without sacrificing the maximum diversity gain and minimum decoding
complexity of the COSTBCs.

\appendices
\section{Single Symbol Decodable Property Of Cascaded Alamouti Code}
\label{ssdcascala}
In this section of the Appendix we show that COSTBCs (cascaded Alamouti code, Example \ref{codecascala})
have the single symbol decodable property for a $N$-hop network when $M_0=\ldots= M_{N-1}=2$.
We first establish this for $N=2$ and then generalize it to arbitrary $N$
using mathematical induction.

To construct COSTBC for $N=2, M_0=M_1=2$, let $\bS_0$ be the Alamouti code
which is given by:
\[\left[\begin{array}{cc}
s_1 & s_2 \\
-s_2^{*} & s_1^{*}
\end{array}\right]\]
where $s_1$ and $s_2$ are constituent symbols.
Then the $2\times 1$ received signal at each relay is
\[\left[\begin{array}{c}
r_m^{1} \\
r_m^{2} \end{array}\right] = \sqrt{E_0}\left[\begin{array}{cc}
s_1 & s_2 \\
-s_2^{*} & s_1^{*}
\end{array}\right]\left[\begin{array}{c}
h_{1m} \\
h_{2m} \end{array}\right] + \left[\begin{array}{c}
n_m^{1} \\
n_m^{2} \end{array}\right]\]
for $m=1,2$.
Transforming
\[\left[\begin{array}{c}
r_m^{1} \\
-r_m^{2*} \end{array}\right] =
\sqrt{E_0}\underbrace{\left[\begin{array}{cc}
h_{1m} & h_{2m} \\
-h_{2m}^{*} & h_{1m}^{*}
\end{array}\right]}_{\tilde{\bH}_m}\left[\begin{array}{c}
s_{1} \\
s_{2} \end{array}\right] + \left[\begin{array}{c}
n_m^{1} \\
-n_m^{2*} \end{array}\right]\]
for $m=1,2$.
Premultiplying by ${\tilde{\bH}^{*}_m}$
\begin{equation}
\label{relaytxala}
\left[\begin{array}{c}
\tilde{r}_m^{1} \\
\tilde{r}_m^{2*} \end{array}\right] =
\sqrt{E_0}\underbrace{\left[\begin{array}{cc}
|h_{1m}|^2+ |h_{2m}|^2 & 0 \\
0 & |h_{1m}|^2+ |h_{2m}|^2
\end{array}\right]}_{{\sfH}_m}\left[\begin{array}{c}
s_{1} \\
s_{2} \end{array}\right] + \underbrace{\left[\begin{array}{c}
n_m^{1}h_{1m}^{*} + n_m^{2*}h_{2m}\\
n_m^{1}h_{2m}^{*} - n_m^{2*}h_{1m} \end{array}\right]}_{\tilde{\bn}_m}
\end{equation}
for $m=1,2$.
Now premultiplying by $\sfH_m^{-\frac{1}{2}}$
\begin{eqnarray}
\label{relaytxala}
\left[\begin{array}{c} \nonumber
\hat{r}_m^{1} \\
\hat{r}_m^{2*} \end{array}\right] &=&
\sqrt{E_0}\underbrace{\left[\begin{array}{cc}
\sqrt{|h_{1m}|^2+ |h_{2m}|^2} & 0 \\
0 & \sqrt{|h_{1m}|^2+ |h_{2m}|^2}
\end{array}\right]}_{{\hat \bH}_m}\left[\begin{array}{c}
s_{1} \\
s_{2} \end{array}\right] \\
&&+ \underbrace{\frac{1}{\sqrt{|h_{1m}|^2+ |h_{2m}|^2}}\left[\begin{array}{c}
n_m^{1}h_{1m}^{*} + n_m^{2*}h_{2m}\\
n_m^{1}h_{2m}^{*} - n_m^{2*}h_{1m} \end{array}\right]}_{\hat{\bn}_m}
\end{eqnarray}
for $m=1,2$.
It is easy to check that the entries of vector ${\hat \bn}_m$ are
 ${\cal CN}(0,1)$ distributed and uncorrelated with each
other from which it follows that entries of vector ${\hat \bn}_m$ are independent.
Using
\[\bA_1 = \left[\begin{array}{cc}
1 & 0 \\
0 & 1 \end{array}\right], \bB_1 = {\bf 0}_2,\
 \bA_2 = {\bf 0}_2,\
\bB_2 = \left[\begin{array}{cc}
0 & -1 \\
1 & 0 \end{array}\right]\]
the $2\times 1$ transmitted signal
from relay $1$ and $2$ is given by
\begin{eqnarray}
\label{cascalatx}
\bt_1 &=& \theta_1\left[\begin{array}{c}
\sqrt{|h_{11}|^2+ |h_{21}|^2}s_1 \\
\sqrt{|h_{11}|^2+ |h_{21}|^2}s_2 \end{array}\right] +
\frac{\theta_2}{\sqrt{|h_{11}|^2+ |h_{21}|^2}}\left[\begin{array}{c}
n_1^{1}h_{11}^{*} + n_1^{2*}h_{21}\\
n_1^{1}h_{21}^{*} - n_1^{2*}h_{11}
 \end{array}\right]\nonumber\\
\bt_2 &= &\theta_1\left[\begin{array}{c}
-\sqrt{|h_{12}|^2+ |h_{22}|^2}s_2^{*} \\
\sqrt{|h_{12}|^2+ |h_{22}|^2}s_1^{*} \end{array}\right] +
\frac{\theta_2}{\sqrt{|h_{12}|^2+ |h_{22}|^2}}\left[\begin{array}{c}
-(n_2^{1}h_{22}^{*} -    n_2^{2*}h_{12})^{*}\\
(n_2^{1}h_{12}^{*} + n_2^{2*}h_{22})^{*}
 \end{array}\right],
\end{eqnarray}
where $\theta_1$ and $\theta_2$ are
the scaling factors so that the power transmitted by each relay is $E_1$,
$\bbE\bt_m^*\bt_m=E_1M_1,\ m=1,2$.
Recall from the COSTBC construction that
$\bS_1 \bydef  [\bA_1\bs+\bB_1\bs^{\dag} \ \  \bA_2\bs+\bB_2\bs^{\dag}]$, where
$\bs = [s_1,s_2,\ldots,s_L]$ is the vector of the constituent symbols of
$\bS_0$. In this case $\bs = [s_1,s_2]$ and $\bS_1$ is
\[ \left[\begin{array}{cc}
s_1 & -s_2^{*} \\
s_2 & s_1^{*}
\end{array}\right]\] which is the Alamouti code and hence an OSTBC as
required.

The $2\times 1$ received signal at the $j^{th}$
receive antenna of the destination is given by
\begin{eqnarray*}
\left[\begin{array}{c}y_j^1 \\
y_j^2\end{array}\right] & = & \theta_1\left[\begin{array}{c}
g_{1j}\sqrt{|h_{11}|^2+ |h_{21}|^2}s_1 -g_{2j}\sqrt{|h_{12}|^2+ |h_{22}|^2}s_2^{*} \\
g_{1j}\sqrt{|h_{11}|^2+ |h_{21}|^2}s_2 +g_{2j}\sqrt{|h_{12}|^2+ |h_{22}|^2}s_1^{*}
\end{array}\right] \\
 & & + \theta_2\left[\begin{array}{c}\frac{g_{1j}(n_1^{1}h_{11}^{*} + n_1^{2*}h_{21})}{\sqrt{|h_{11}|^2+ |h_{21}|^2}}  - \frac{g_{2j}(n_2^{1}h_{22}^{*} - n_2^{2*}h_{12})^{*}}{\sqrt{|h_{12}|^2+ |h_{22}|^2}} \\
\frac{g_{1j}(n_1^{1}h_{21}^{*} - n_1^{2*}h_{11})}{\sqrt{|h_{11}|^2+ |h_{21}|^2}} + \frac{g_{2j}(n_2^{1}h_{12}^{*} + n_2^{2*}h_{22})^{*}}{\sqrt{|h_{12}|^2+ |h_{22}|^2}}\end{array}\right] +
\left[\begin{array}{c}z_j^1 \\ z_j^2\end{array}\right]
\end{eqnarray*}
for $j=1,2$.
We denote
$\eta_1 = \frac{(n_1^{1}h_{11}^{*} + n_1^{2*}h_{21})}{\sqrt{|h_{11}|^2+ |h_{21}|^2}}$,
$\eta_2 = \frac{(n_1^{1}h_{21}^{*} - n_1^{2*}h_{11})}{\sqrt{|h_{11}|^2+ |h_{21}|^2}}$,
$\eta_3 = \frac{(n_2^{1}h_{22}^{*} - n_2^{2*}h_{12})}{\sqrt{|h_{12}|^2+ |h_{22}|^2}}$
$\eta_4 = \frac{(n_2^{1}h_{12}^{*} + n_2^{2*}h_{22})}{\sqrt{|h_{12}|^2+ |h_{22}|^2}}$.
Note that ${\bbE}\eta_i\eta_j^{*}= 0,
\ \forall i,\ j =1,2,3,4 \ i\ne j$.
Rewriting,
\begin{eqnarray*}
\left[\begin{array}{c}y_j^1 \\
y_j^{2*}\end{array}\right] & = & \theta_1\underbrace{\left[\begin{array}{cc}
g_{1j}\sqrt{|h_{11}|^2+ |h_{21}|^2} &  -g_{2j}\sqrt{|h_{12}|^2+ |h_{22}|^2} \\
 g_{2j}^{*}\sqrt{|h_{12}|^2+ |h_{22}|^2} & g_{1j}^{*}\sqrt{|h_{11}|^2+ |h_{21}|^2}
\end{array}\right]}_{\Phi_j}\left[\begin{array}{c}s_1 \\
s_2^{*}\end{array}\right] \\
 & & + \theta_2\left[\begin{array}{c}g_{1j}\eta_1 - g_{2j}\eta_3^{*} \\
(g_{1j}\eta_2 + g_{2j}\eta_4^{*})^{*}\end{array}\right] +
\left[\begin{array}{c}z_j^1 \\ z_j^{2*}\end{array}\right]
\end{eqnarray*}
Denoting ${\tilde h}_1 = |h_{11}|^2+ |h_{21}|^2$ and
${\tilde h}_2 = |h_{12}|^2+ |h_{22}|^2$,
premultiplying by $\Phi_j^{*}$, it follows that
\begin{eqnarray*}
\Phi_j^{*}\left[\begin{array}{c}y_j^1 \\
y_j^{2*}\end{array}\right] & = & \theta_1\left[\begin{array}{cc}
(|g_{1j}|^2{\tilde h}_1 + |g_{2j}|^2{\tilde h}_2)s_1\\
(|g_{1j}|^2{\tilde h}_1 + |g_{2j}|^2{\tilde h}_2)s_2\end{array}\right]\\
 & & + \theta_2
\underbrace{\Phi_j^{*}\left[\begin{array}{c}g_{1j}\eta_1 - g_{2j}\eta_3^{*} \\
(g_{1j}\eta_2 + g_{2j}\eta_4^{*})^{*}\end{array}\right] +
\Phi_j^{*}\left[\begin{array}{c}z_j^1 \\ z_j^{2*}\end{array}\right]
}_{\bv_j}
\end{eqnarray*}
Expanding $\bv_j$, we have

\[\bv_j = \theta_2 \left[\begin{array}{c}
|g_{1j}|^2\sqrt{{\tilde h}_1}\eta_1 - g_{1j}^{*}g_{2j}\sqrt{{\tilde h}_1}\eta_3^{*} +
 g_{1j}^{*} g_{2j} \sqrt{{\tilde h}_2}\eta_2^{*} +  |g_{2j}|^2\sqrt{{\tilde h}_2}\eta_4
 \\
- g_{1j}g_{2j}^{*}\sqrt{{\tilde h}_2}\eta_1 + |g_{2j}|^2\sqrt{{\tilde h}_2}\eta_3^{*}+
  |g_{1j}|^2\sqrt{{\tilde h}_1}\eta_2^{*} +  g_{1j}g_{2j}^{*}\sqrt{{\tilde h}_1}\eta_4
\end{array}\right] +
\left[\begin{array}{c}
g_{1j}^{*}\sqrt{{\tilde h}_{1}}z_j^1 + g_{2j}^{*}\sqrt{{\tilde h}_{2}}z_j^{2*} \\
 -g_{2j}^{*}\sqrt{{\tilde h}_{2}}z_j^1 + g_{1j}^{*}\sqrt{{\tilde h}_{1}}z_j^{2*}
\end{array}\right]  \]
%
It is easy to check that ${\bbE}\bv_{j1}\bv_{j2}^{*} = 0$,
${\bbE}\sum_{j=1}^{M_2}\bv_{j1}\left(\sum_{j=1}^{M_2}\bv_{j2}\right)^{*} =0$ and $\bv_{ji},\ i=1,2$ is circularly symmetric complex Gaussian
which implies that $\sum_{j=1}^{M_2}\bv_{j1}$ and $\sum_{j=1}^{M_2}\bv_{j2}$
are independent and thus
both $s_1$ and $s_2$ can be decoded independently of each other without any
loss in performance compared to joint decoding. Thus we conclude that
cascaded Alamouti code has the single symbol decodable property for $N=2,M_0=M_1=2$.

To extend this result to the $N$-hop case we use mathematical
induction where $M_n=2 \ \forall n=0,1\ldots,N-1$. We have shown the
result for $N=2$, thus we can start the induction. Let us assume
that the result is true for $k$-hop network. From the induction
hypothesis, the cascaded Alamouti code has the single symbol decodable property for $k$-hop
network, which means that at the $j^{th}$ receive antenna $j=1,2$ of
the destination of $k$-hop network, using CSI, the received signal
$\by_j$ can be transformed into ${\hat \by}_j$, where
\[{\hat \by}_{j} = \alpha\left[\begin{array}{c}
c_js_1 \\
c_js_2 \end{array}\right] +
\beta\left[\begin{array}{c}
z_{j1}\\
z_{j2} \end{array}\right]\]
$c_j$ is the channel gain, $\alpha$ and $\beta$
are the scaling factors and $z_{j1},z_{j2}$ are noise terms which are complex
Gaussian distributed with zero mean and $\sigma^2_k$ variance
and are independent of each other.
We extend the $k$-hop network to $k+1$ network
by assuming that the actual destination is one more hop away and using the
destination of the $k$-hop network as the $k^{th}$ relay stage with $2$
relays with a single antenna each.
Then using
\[\bA_1 = \left[\begin{array}{cc}
1 & 0 \\
0 & 1 \end{array}\right], \bB_1 = {\bf 0}_2,\
 \bA_2 = {\bf 0}_2,\
\bB_2 = \left[\begin{array}{cc}
0 & -1 \\
1 & 0 \end{array}\right]\] at the two relays of the $k^{th}$ relay stage,
the $2\times 1$ transmitted signals $\bt_1$ and $\bt_2$
from the relay $1$ and $2$ of the $k^{th}$ relay stage,
respectively,  are given by
\[\bt_1 = {\hat \alpha}\left[\begin{array}{c}
c_1s_1 \\
c_1s_2 \end{array}\right] +
{\hat \beta}\left[\begin{array}{c}
z_{11}\\
z_{12} \end{array}\right]\]
\[
\bt_2 = {\hat \alpha}\left[\begin{array}{c}
-c_2s_2^{*} \\
c_2s_1^{*} \end{array}\right] +
{\hat \beta}\left[\begin{array}{c}
-z_{21}^{*}\\
z_{22}^{*}
 \end{array}\right]
\]
where ${\hat \alpha}$ and ${\hat \beta}$ are such that $\bbE\bt_j^*\bt_j=
E_kM_k$.
Recall that these transmitted signals are similar to the transmitted
signals by cascaded Alamouti code in the $N=2$ case (\ref{cascalatx}) where
$c_1 =\sqrt{|h_{11}|^2+ |h_{21}|^2}$,
$c_2=\sqrt{|h_{12}|^2+ |h_{22}|^2}$,
$z_{11} = \frac{n_1^{1}h_{11}^{*} + n_1^{2*}h_{21}}{\sqrt{|h_{11}|^2+ |h_{21}|^2}}$,
$z_{12}=\frac{n_1^{1}h_{21}^{*} - n_1^{2*}h_{11}}{\sqrt{|h_{11}|^2+ |h_{21}|^2}}$,
$z_{21} = \frac{n_2^{1}h_{22}^{*} - n_2^{2*}h_{12}}{\sqrt{|h_{12}|^2+ |h_{22}|^2}}$,
$z_{22}  = \frac{n_2^{1}h_{12}^{*} + n_2^{2*}h_{22}}{\sqrt{|h_{12}|^2+ |h_{22}|^2}}$.
Using similar arguments as in the $N=2$ case, it easily follows that
cascaded Alamouti code has the single symbol decodable property for $k+1$-hop network from which we
can conclude that cascaded Alamouti code has the single symbol decodable property for arbitrary
$N$-hop networks with $M_0=\ldots=M_{N-1}=2$.
In the next section we show that COSTBCs have the
single symbol decodable property for an arbitrary $N$-hop network.

\section{Single Symbol Decodable Property Of COSTBC}
\label{ssdcostbc}
In this section we show that COSTBCs have the single symbol decodable property. We first show this
for $2$-hop networks and then generalize it to $N$-hop networks where $N$
is any arbitrary integer.
Let $\bS_0$ be the transmitted OSTBC from the source and
 $\bs = [s_1, \ldots, s_L]^T$ be the vector of the constituent symbols of $\bS_0$. Then from (\ref{scaledssd}), using CSI,
the received signal $\br^1_k$ at the
$k^{th}$ relay of relay stage $1$ can be transformed into ${\hat \br}^1_k$ where
\begin{equation*}
{\hat \br}_k^1 =  \sqrt{E_0}
\underbrace{\left(\begin{array}{cccc}
\sqrt{\sum_{m=1}^{M_0}|h_{mk}|^2} & 0 & 0 & 0 \\
0 & \sqrt{\sum_{m=1}^{M_0}|h_{mk}|^2}  & 0 & 0 \\
0 & 0 & \ddots & 0 \\
0 & 0 & 0&  \sqrt{\sum_{m=1}^{M_0}|h_{mk}|^2}\\
\end{array}\right)}_{\sfH^{\frac{1}{2}}}\bs + \hat{\bn}_k
\end{equation*}
and the entries of ${\hat \bn}_k$ are independent and
${\cal CN}(0,1)$ distributed.
For $N=2$, from (\ref{recdestant})
the received signal at the $j^{th}$ antenna of the destination can be
written as
\[\by_j = [\bt^1_1 \ \bt^1_2 \ \ldots \ \bt^1_{M_1}]\bg_{j} + \bz_j\] for
$j = 1,2, \ldots M_2$, where $\bt^1_k$ is the transmitted vector from relay $k$
(\ref{relaytx}) of relay stage $1$.
The received signal $\by_j$ can also be written as
\begin{eqnarray*}
\by_j &=& \sqrt{\frac{E_0E_1M}{L\gamma}}\bS_1\left[\begin{array}{c}
\sqrt{\sum_{m=1}^{M_0}|h_{m1}|^2}g_{1j} \\ \sqrt{\sum_{m=1}^{M_0}|h_{m2}|^2}g_{2j} \\ \vdots \\\sqrt{\sum_{m=1}^{M_0}|h_{mM_1}|^2}g_{M_1j}\end{array}\right] \\
&& +\underbrace{\sqrt{\frac{E_1M_1}{L\gamma}}[\bA_1\hat{n}_1 + \bB_{1}\hat{n}^{\dag}_{1} \
\bA_2\hat{n}_2 + \bB_{2}\hat{n}^{\dag}_{2} \ \ldots \
\bA_{M_1}\hat{n}_{M_1} + \bB_{M_1}\hat{n}^{\dag}_{M_1} ]\bg_j +
\bz_j}_{\bw_j}\end{eqnarray*} where
$\bS_1 = [\bA_1\bs+\bB_{1}\bs^{\dag} \ \bA_2\bs+\bB_{2}\bs^{\dag} \  \ldots \
\bA_{M_1}\bs+\bB_{M_1}\bs^{\dag} ]$.

Since $\bS_1$ is an OSTBC, invoking the single symbol decodable
property of OSTBC (\ref{ssd}) and
using the fact that entries of $\bw_j$ are independent, it follows that,
using CSI, the received signal $\by_j$ can be transformed into ${\hat \by}_j$,
where
\[{\hat \by}_j =  \sqrt{\frac{E_0E_1M_1}{L\gamma}}
\left(\begin{array}{ccc}
\sum_{k=1}^{M_1}|g_{kj}|^2\left(\sum_{m=1}^{M_0}|h_{mk}|^2\right) & 0 & 0 \\
0 & \ddots & 0 \\
0 & 0 & \sum_{k=1}^{M_1}|g_{kj}|^2\left(\sum_{m=1}^{M_0}|h_{mk}|\right)^2\\
\end{array}\right)\bs + \hat{\bw}_j
\]
and the entries of $\hat{\bw}_j$ are independent.
Thus, it is clear that all the constituent symbols $s_1, \ldots, s_L$ can be separated
with independent noise terms and  we conclude that COSTBCs have the
single symbol decodable property for a $2$-hop network.
Using mathematical induction, similar to the Appendix \ref{ssdcascala},
it can be easily shown that COSTBCs also have the single symbol decodable
property for arbitrary $N$-hop network and for brevity we omit
it here.

\bibliographystyle{IEEEtran}
\bibliography{IEEEabrv,Research}

\end{document}